\documentclass[manuscript,screen,authorversion]{acmart}

\AtBeginDocument{%
  }

\setcopyright{acmlicensed}
\copyrightyear{2018}
\acmYear{2018}
\acmDOI{XXXXXXX.XXXXXXX}

\acmISBN{978-1-4503-XXXX-X/2018/06}

 \usepackage{tikz,calc}
 \usetikzlibrary{decorations.markings}
 \usepackage{soul}
 \setul{}{0.75pt}
 \setuloverlap{1.5pt}

\usepackage{mwe}
\usepackage{lscape,booktabs,longtable}
\usepackage{amsmath}
\usepackage{comment}
\usepackage{enumerate}
\usepackage[ruled,vlined]{algorithm2e}
\usepackage[graphicx]{realboxes}
\usepackage{todonotes}
\usepackage{tabularx}

\newcommand{\rev}[1]{\textcolor{black}{#1}}

\usepackage{subcaption}

\usepackage{cleveref}

\usepackage{outlines}
\usepackage{tikz}
\usepackage{url}
\usepackage{xcolor}
\allowdisplaybreaks
\usetikzlibrary{arrows, positioning, shapes}
\tikzstyle{process} = [rectangle, minimum width=3cm, minimum height=1cm, text centered, draw=black, fill=white]
\tikzstyle{queue} = [circle, minimum size=1cm, text centered, draw=black, fill=white]
\tikzstyle{arrow} = [thick,->,>=stealth]

\begin{document}

\title{Due Process on Hold:
A Queueing Framework for Improving Access in SNAP}

\author{Andrew Daw}
\authornote{Alphabetical order}
\authornotemark[0]
\email{dawandre@usc.edu}
\affiliation{%
  \institution{University of Southern California}
  \country{USA}
}
\author{Chloe Pache}
\authornotemark[0]
\affiliation{%
  \institution{University of California, Santa Barbara}
    \country{USA}
}

\author{Angela Zhou}
\authornotemark[0]
\email{zhoua@usc.edu}
\affiliation{%
  \institution{University of Southern California}
    \country{USA}
}


\begin{abstract}
The U.S. social safety net delivers essential services at mass scale, but access burdens persist, as congested contact or call centers serve as a primary mode of application completion and assistance. In Holmes v. Knodell, Missouri’s SNAP call centers were so congested that nearly half of all application denials were procedural—caused by applicants’ inability to complete required interviews, rather than underlying ineligibility. The judge ruled these system failures led to a violation of procedural due process.
We propose a performance evaluation framework based on queueing models from operations research and management to assess and improve access in such systems. Operational access failures of call centers are distinct from prior automation failures in benefits provision. \textit{Emergent arbitrariness} arises from interactions between system dynamics and access demand, rather than from an explicit algorithmic rule, making diagnosis and repair inherently system-level. We develop a queueing model that incorporates phenomena that distinguish social services from standard service domains, redials and abandonment, through which backlogs generate \textit{endogenous congestion}. Standard queueing guidance from Erlang-A that does not address endogenous congestion fundamentally understaffs, which could lead to persistent shortfalls in practice. Using a fluid approximation, we derive steady-state performance metrics to analytically characterize the impacts of bundled staffing and service delivery changes. We fit model parameters to call-center data disclosed in court documents. Our queueing model can support ex-ante evaluation and design of access systems, inform policy levers for improving access, and provide evidence about whether applicants are afforded a meaningful opportunity to be served at scale.
\end{abstract}



\maketitle

\section{Introduction}




Social safety net programs deliver essential services and benefits at a massive scale under major resource constraints. Operational challenges in service delivery can undermine efforts to achieve policy goals and program effectiveness. Recently, a district court case, \textit{Holmes v. Knodell}, ruled that Missouri Department of Social Services’ (DSS) Supplemental Nutrition Assistance Program (SNAP) call center was so congested that enrollees were unable to call through to complete required interviews and applications, constituting a procedural due process violation.
\textit{Procedural denials} occur when an application is denied purely on procedural or administrative grounds, such as missing an interview, rather than the eligibility facts of the case. In Missouri, such procedural denials were as high as 50\% of all denials. 
To the best of our knowledge, this is the first time that an \textit{operational system (the call center) was found to violate procedural due process}. While the case ruling has concluded, the story is not complete. Missouri DSS has been ordered to improve call center operations under severe staffing constraints, reflecting a broader challenge across public-sector service delivery: how to provide meaningful access at scale with limited resources. How should Missouri DSS improve their call center design to provide \textit{de facto} access to complete \textit{de jure} application requirements?

The operational failures of long wait times in automated systems are systemic both in \textit{origin} and \textit{future improvement}, differing from prior automation failures in benefits provision. We do not find available technical guidance that addresses the specific challenges of social service delivery. 
Queueing theory, which is a key subdiscipline of operations research, is the natural tool for performance analysis. Call centers have been a particularly fruitful domain, as they feature the core ingredients of the field: large volumes of customers, random fluctuations in the times between arrivals and the durations of service, and strong potential for long wait times despite the risks of impatient customers. 
Indeed, federal guidelines for performance management of SNAP call centers refer to queueing theory and mention the famous Erlang-A formula, while also conceding that the \textit{specific} nature of social services means that standard off-the-shelf models may not directly apply; hence, the assumptions for the Erlang-A formula do not actually hold in practice \citep{fns}. 

Here, we develop a queueing model to fill the gap between standard call center models and special challenges in social services. 
Call centers for SNAP, and benefits more broadly, specifically suffer from \textit{abandonment} where busy callers, some of whom may be paying by the minute \citep{howe_new_2025}, cannot stay on hold forever and may drop their call. But, the key difference from standard models is that callers applying to satisfy SNAP interview requirements who abandon are likely to call back, or \textit{re-dial}, since SNAP and other social safety net programs serve urgent needs without private outside options. Moreover, interviews often surface additional follow-up steps resulting in future calls. These behaviors, distinctive to social services, \textit{compound operational deficiencies}, introducing stark feedback loops where prior backlogs result in abandonment that only worsens future congestion due to re-dials. Our model captures these distinctions and is therefore relevant beyond SNAP alone, for example for backlogs for unemployment insurance or Social Security Administration (SSA) call centers, which have suffered major backlogs and similar high-abandonment issues at times including prior economic crises \citep{pahlka2023recoding,Coffey_2017,ho2025evaluation} and current understaffing \citep{Rain_Natanson_25AD}.


\rev{We focus on a fluid approximation of the queueing model, which is a deterministic approximation to the queueing system under large arrival volumes and with many servers, both of which are evident in the  SNAP call center data from \emph{Holmes v. Knodell} \citep{holmesKnodellPlaintiffsMSJ2024}. Agencies could use our framework to understand trade-offs between different call center system design choices and resource expenditures. Our analysis in the fluid model enables us to vary system parameters and answer questions such as  ``if I can only hire 10 more staff, by how much would I have to reduce handling time in order to reduce wait time below 20 min?'' To preview how this can be useful, first we note that DSS estimated that they needed at least 150 more staff\footnote{\rev{DSS had 200-300 staff split among different tasks including handling calls and scheduled interviews, but the agency had earlier estimated that they needed >400 staff dedicated to answering calls alone to answer 80\% of calls within 2 minutes \citep{holmesKnodellPlaintiffsMSJ2024}.}}, most likely using the Erlang A formula, which is the only the only operations guidance we see in federal documents \citep{fns}. But DSS already faced difficulties with retention, let alone hiring and onboarding new staff which takes a long time and may not grant immediate relief. Our quantitative model shows the same performance improvement can be achieved by \textit{bundling operational changes} with fewer additional staff. In SNAP, the Erlang-A's model assumption of permanent abandonments fundamentally fails: applying callers don't have outside options, so abandoned calls turn into re-dials, and therefore congestion itself generates future arrivals. Our finer-grained model includes the \textit{endogenous congestion} generated by re-dials and includes key dynamics that arise from program design and implementation. Accordingly, using our model, policymakers can directly assess \textit{how changes in program implementation and operations} (such as changes in timeframe for obtaining follow-up documentation) affect system performance. }

The focus of our work is in developing a performance evaluation and improvement framework, expanding upon classical queueing models from operations research, that can support efforts to  improve access in social services. Our model, though reliant on some probabilistic assumptions, allows us to estimate \rev{six key performance metrics: the mean number of waiting callers, the mean waiting time, the average speed to answer (i.e., the waiting time for callers who don't abandon), the procedural denial rate, the endogenous congestion from re-dials, and the endogenous congestion from re-certification. The first three waiting-based metrics are hallmark measures of performance in essentially any call-center, and the latter three are specifically relevant in this SNAP call center setting, where the dynamics of the benefits certification process creates patterns of returning callers that can magnify the load on the system and are not typical in other call center settings.}
The queueing model captures the \textit{feedback loops} that arise when relying on call centers for required interviews and re-certifications that themselves generate follow-up contacts. 
We introduce our queueing model, its fluid approximation, and discuss how to calibrate the model to data, which we do using court documents from \textit{Holmes v. Knodell}. We introduce a simple dashboard to explore the impact of design changes on performance metrics. Our model admits analytical expressions for \rev{the six focal performance metrics}, which we analyze to obtain insights on how different design changes impact performance. \rev{Furthermore, the dashboard demonstrates how the model also lets us evaluate many other quantities beyond the six highlighted in the paper.} Finally, we use our calibrated model to illustrate how an agency might combine staffing and system/policy design levers, like reducing average handling time or lengthening re-certification cycles, in order to achieve greater improvement with less staffing. 
\section{Background}
\paragraph{Background on SNAP}
The Supplemental Nutrition Assistance Program (SNAP), formerly the Food Stamp Program, is the largest means-tested social assistance program in the United States and provides monthly food-purchasing benefits to low-income households. SNAP is federally authorized and funded but administered by state and local agencies, which are responsible for determining eligibility and operating application and recertification processes within federal rules. In fiscal year 2024, ``SNAP served an average of approximately 41.7 million individuals per month—about 12.3\% of the U.S. population" \citep{ers_snap_key_stats}.\footnote{Participation by state varies widely, from as high as 21.2\% in New Mexico to 4.8\% in Utah \citep{ers_snap_state_variation}.} Studies of SNAP and other benefits programs find improvements not only in food security, but also health and labor outcomes, as well as stabilization against macroeconomic downturns \citep{hoynes2015us}. Nonetheless, many who are eligible for SNAP do not receive benefits. Program eligibility is complex, including income determination and verification, deductions, asset limits, household definition, and work requirements --- as is the application \citep{crs_snap_r42505}. Such complexity introduces ``administrative burden,"  the requisite learning, psychological load, and compliance costs that citizens experience in interactions with government \citep{herd2025administrative}. 
Many who are denied benefits because of operational frictions (missed interviews, application errors) are actually eligible \citep{homonoff2021program}. 

\textit{The role of the interview, applications and recertification process: } Call and contact centers remain one of the primary ways that potential enrollees interact with SNAP program administration \citep{fns}.\footnote{Services provided by call centers can provide everything from general information to official case services like updating and processing changes, conducting interviews, providing updates on processing, appointment scheduling and more. Staffing SNAP call centers is difficult because of SNAP policy complexity and the need to access live case information; call center staff touching the eligibility process must be merit system personnel \citep{fns}. As a result, there is large turn-over and chronic under-staffing, and it can take a long time to onboard new staff  \citep{Coffey_2017} --- Missouri DSS faced these challenges as well.} 
Key administrative \textit{checkpoints} \citep{kim2025administrative} in the SNAP process include initial application, which requires an interview, and re-certification every 6 to 12 months, which also requires an interview. The re-certification form measures changes of circumstances (including costs/income, household composition, etc.) 
which may change eligibility; any changes require submitted verification. During interviews, the caseworker may ask for additional documents to prove certain eligibility requirements or provide additional detail. Initial interviews also often serve as opportunities to inform applicants about complex SNAP-specific definitions. 

Missouri DSS has a waiver from Food and Nutrition Services (FNS), the federal agency administering SNAP, to implement interview procedures deviating from federal guidelines - namely, DSS does not need to schedule interviews, but must provide applicants with a notice (interview letter) instructing them to contact the Call Center within 5 days of submitting their application. If they do not, they receive a notice of missed interview informing that they need to complete the interview within 30 days of submitting their application. 

More broadly, there are no uniform performance standards for operational performance for interviews.\footnote{\citet{howe_new_2025} comment on the legal context of \textit{Holmes v. Knodell} and the challenge of rendering on-the-ground wait times legally cognizable. Existing enforceable laws did not anticipate the omnichannel service and access needs of today's day and age.} Although interviews are \textit{required} by federal regulation to be completed within 30 days, federal regulation does not stipulate performance standards on service quality of \textit{how} interview opportunities should be granted to enrollees. It remains unclear what \textit{exact} metrics might lead to judicial findings of procedural due process violations.\footnote{For example, the judge noted that in Missouri near 50\% of denials were procedural. But other jurisdictions also have high rates of procedural denials that reflect persistent access barriers; for example, Los Angeles in 2018-2019 had a baseline rate of $\sim30$\% procedural denials \citet{giannella2024administrative}.} 

\textit{The grounds for the judicial decision:} 
The case facts establish that the call center and administrative process (including low-staffed resource centers) to handle SNAP applications are overwhelmed. Plaintiffs contend that wait times are unacceptable, with some plaintiffs calling ten to twenty times and unable to call through, and legal services advocates waiting in the queue for hours before getting disconnected \citep{holmesKnodellPlaintiffsMSJ2024}. DSS is understaffed.\footnote{DSS has 200-300 staffers split in between informational (non-interview) calls and interview calls, while it estimates it needs 400 interview-only staffers to meet private call center “industry standards” of answering 80\% of calls within 2 minutes} Thirty-two percent of all calls that made it into the interview-only queue were abandoned by the caller. Fifty percent of all applicants were denied for failure to complete the interview in one month during 2023, due only to a failure of the system to offer a reasonable opportunity to interview. 
The judge concludes that DSS was unsuccessful in providing ``timely, accurate and fair service,'' that the ``reliance on an inadequate automated system and understaffed offices to provide interviews also violates Defendant’s obligation under SNAP and Defendant’s on-demand waiver,'' and therefore    denial based on the automated system is a wrongful denial of benefits based on the arbitrariness of whether enrollees can call through \citep{holmesKnodell2024}.\footnote{``Here, applicants are being denied benefits not based on the merits of their application but on the failure to obtain an interview. Further, the failure to interview is a direct result of Defendant's inability to provide an efficient and successful system that allows applicants to schedule and complete an interview within the required time frame.'' \citep{holmesKnodell2024}} 

\textit{Background on procedural due process:} Much prior work in the algorithmic fairness community has focused on connecting algorithmic harms or performance legal theories of discrimination. \textit{Holmes v. Knodell} is a \textit{procedural due process}\footnote{\textit{Procedural due process} requires that the state follow certain procedures when depriving someone of their property, such as giving notice, an opportunity to be heard (for example, in court), and a neutral decision-maker  \citep{ames2020due}. 
Courts have ruled, since the Goldberg v. Kelly, 1970, Supreme Court case, that social safety net and welfare programs represent significant property interests and hence are subject to procedural due process protections. 
Traditionally, such protections, drawing upon due process in courtroom procedures, might take the form of adversarial hearings, appeals, tribunals, or legal representation. } case, like other cases of algorithmic harms in benefits provision.  
Scholars in administrative law have argued that such post-hoc appeals and hearings may not provide strong enough procedural protections for the massive scale of social safety net programs. \citet{mashaw1973management} studies the Social Security Disability Insurance (SSDI) program and adjudication process and argues for the ``management side of due process.'' 
Mashaw argues that burdensome appeals surface a small fraction of the agency's actual errors. Instead, fairness, accuracy, and timeliness are properties of \textit{system performance}, not just rare appeals. Pro-active \textit{continuous monitoring and performance management} can provide more systematic guarantees than individualized procedure that is overly reliant on ex-post correction. {Fairness and accurate determinations are therefore emergent property of a system's performance as a whole, not necessarily individual procedural protections}. 
\citep{ho2017managing,ames2020due} connect these arguments about performance management and quality assurance programs in agencies 
with modern calls for AI auditing, algorithmic accountability, and performance management, including \citep{ho2025evaluation} an example audit of benefits-assistance chatbots for UI assistance.

We share a similar high-level emphasis on \textit{the role of performance evaluation in operationalizing due process in benefits provision}. However, we focus specifically on call center performance, which introduces distinct technical challenges. 






\paragraph{Emergent systemic arbitrariness from automated systems.}
Automated systems have introduced great harms in crucial social safety net programs and benefits provision \citep{eubanks2018automating,btah,richardson2019litigating}, due to flawed data, problematic system design choices, or inherent system limitations \citep{gules2025algorithmic}. 
\textit{Holmes v. Knodell} introduces new challenges because the problematic service deficiencies, like long wait times and inability to reach a live representative, are systemic consequences of the automated system, rather than a singular ``incorrect" algorithmic rule or implementation quirk that clearly violates procedural due process protections. 
We call this \textit{emergent arbitrariness}, which differs from other algorithmic failures in benefits teach where algorithmic rules directly violate \textit{de jure} procedural due process (e.g., defective notice or fully automated determinations).
In a referenced prior case as reference for arbitrary agency action, Barry v. Lyon, an added automated match imposed an extra eligibility requirement beyond those federally mandated, violating the SNAP Act. Yet, the fix is simple --- remove the automated match. In contrast, for SNAP call centers, the \textit{origins of} and \textit{future improvements} for the \textit{emergent arbitrariness} are more complex than a single system design flaw. The judge can mandate that Missouri DSS must modify its administration process to ensure the opportunity to interview, but it remains unclear \textit{how}.

\section{Related Work}



\textit{Informal and formal performance evaluation of government benefits systems:}  
In formal quality control processes for SNAP, FNS collects random samples of household case files and analyzes eligibility determinations and overpayments/underpayment errors \citep{snapqc}. However, these formal performance evaluations focus on payment errors, rather than other system performance measures including administrative burdens, access, or other measures of otherwise-eligible households that nonetheless do not complete applications or receive benefits. Regarding call center performance specifically, there is wide heterogeneity in evaluation structure and specific metrics used across different government agencies and programs. (See \citet{Coffey_2017} for an overview for unemployment insurance\footnote{The Department of Labor contracted a consulting firm to survey different state leaders about their operational practices in call centers for unemployment insurance. They found wide variation in exact operational practices, but common use of strategies like cross-training and pooling queues across multiple call centers, while finding that many states could use more forecasting and queue management tools \citep{Coffey_2017}.}). However, the IRS has long provided assistance to tens of millions of Americans by phone, and has long designed its own call center quality and customer experience evaluations \citep{rosage2004evolution}.  
In the absence of formal performance assessments and real-time reporting, outsiders may also run informal audits to report performance metrics\footnote{Earlier iterations of the IRS' call center evaluation included "mystery-shopper" like programs, with the downside that queries are artificial rather than actual customer concerns \citep{rosage2004evolution}.} \citep{dube2021note}.
Sen. Elizabeth's Warren's office made hourly calls to SSA, finding wait times of 100 minutes and that only 50\% of calls connected to a live representative \citep{warren}. Specifically in the context of \textit{administrative burdens and SNAP}, recent works study changes in interview operations that introduce or improve administrative burdens, demonstrating beneficial impacts of interview flexibility on approvals and long-term participation \citep{giannella2024administrative,kim2025administrative} or barriers from less time to reschedule \citep{homonoff2021program}.\footnote{\citet{homonoff2021program} find that later recertification interviews, with less time for rescheduling, decrease recertification success, including for truly eligible clients. 
\citet{giannella2024administrative} study a field experiment of applicant-initiated \textit{flexible} on-demand interviews in Los Angeles, increasing approvals by $6\%$ and long-term participation by $2\%.$ \citet{kim2025administrative} run a field experiment comparing text-message reminders of flexible interviews vs. mailed reminders in Boulder County, Colorado, finding earlier and 10\% more likely interview completion.  }

\textit{Other operational analysis of public-sector queueing systems:} 
Queueing analysis is a mature discipline, and prior studies have also focused on public sector systems, albeit not the SNAP access system here. Some prior studies focus on operations in the judicial case management system. \citet{freund2024dedicated} study the queueing system in immigration courts and analyze fairness/efficiency properties of the queueing discipline, as well as strategic considerations specific to when waiting provides utility \citep{freund2025regulating}. \citet{bakshi2025service} study delays in judicial case management. Judicial case management differs from our eligibility/administrative burden setting: in court docket management, cases stay ``live'' for extended periods of time, while for social services, the interview requirement must be completed within 30 days. \citet{anunrojwong2023information} study information design in a general model to target high-need individuals without outside options to manage a congested queue for social services. 

Although classical literature on ordeal mechanisms \citep{nichols1982targeting} posits that ordeals serve a targeting function to dissuade those who are ineligible, the reality of the situation is that administrative burden \citep{herd2025administrative} often prevents actually-eligible individuals from taking up beneficial services \citet{homonoff2021program}. \citep{garg2025heterogeneous} outlines empirically documented gaps in participation and access in other mechanisms.Our recommendations focus on bundled staffing and service delivery interventions that can mitigate current severe blockages in access, rather than achieving finer-grained targeting, which is an interesting question for future work. 

\rev{\textit{Welfare implications of queueing-mediated service provision:} Many service and health interventions in general experience capacity constraints, though this may not be explicitly modeled in explicit system or causal analyses of performance \citep{miles2019causal,boutilier2024operational}, introducing potential social welfare risks in scale-up of otherwise beneficial social interventions. \citet{liu2024redesigning} study service-level agreements under exogenous demand arrivals in municipal operations in a queuing model.
}

 \textit{Other technical work and benefits provision:} \citet{pahlka2023recoding} provides an overview of civic tech and digital services in the federal government, including key benefits tech infrastructure. \citet{escher2020exposing} audits benefits \textit{eligibility calculators} for potential errors. \citet{jo2025ai} explores AI and trust in LLM chatbots supporting SNAP applicants. \citet{koenecke2023popular} audits ad allocation budgets for SNAP outreach and in a survey, finds broad support for equity considerations. Recent evaluations of LLM and chatbots in social services include an audit of UI chatbot support \citep{ho2025evaluation} and \citet{gosciak2026llms} evaluates LLM-chatbot assistance for caseworkers in a randomized-controlled trial, finding that such chatbots can support caseworkers in answering factual questions when they are accurate, but optimal AI-human assistance strategies remain crucial. Recent interest in AI surfaces new opportunities for service redesign that could result in shorter service times or self-service channels, and can be represented in our framework as changes in model parameters.

\section{Queueing Model}




Though the conventional approach to queueing models of call centers does recognize that customers may leave the queue before reaching a service agent, they typically assume that abandoning customers are lost forever \citep{gans2002telephone}. Similarly, customers who complete service are typically thought of as departing the system once and for all. 
For instance, the canonical queueing model of customer abandonment, the Erlang-A, supposes that every customer has an exponentially distributed patience time, such that, if their wait time reaches this patience point, they abandon the queue and leave the service entirely \citep{mandelbaum2005palm,mandelbaum2007service}. While it is certainly not desirable for a customer to abandon, the only ``cost'' paid per abandonment in this model is the lost customer; in fact, on the aggregate, the Erlang-A inherently views abandonment as a panacea for an overwhelmed system, offering stability for a service that otherwise could not keep up with demand. 

Unfortunately, these assumptions are not well-aligned with the SNAP context. Because SNAP is an essential service and prospective enrollees must first complete the interview process, customers who abandon the queue are very likely to call back later. Moreover, because the benefits require regular re-certification, successfully served customers are due to eventually return. Hence, the ``cost'' of abandonment compounds for SNAP call centers: there is the same negative service experience upon abandonment as is modeled by the Erlang-A, but, moreover, the case data shows that the majority of these abandoning callers will soon call back. Hence, unlike what is modeled by the Erlang-A,  abandonment is both an undesirable immediate outcome and an \emph{endogenous driver of congestion} over the long-run. Not only are these real-world features of the SNAP context missing from the Erlang-A model, the Erlang-A will \emph{fundamentally under-staff} system that have endogenous congestion (and we will demonstrate this in our case study on SNAP data; see Table~\ref{tab:compact_staffing_benchmark}). 

With those dynamics in mind, we build from prior work on \emph{re-dials} and \emph{re-connects} in \citet{ding2015fluid} to capture four different possible paths, or ``orbits,'' in which callers may return to the SNAP call center again after abandoning or completing service, with two ``orbits'' from abandonment and two ``orbits'' from completed service.\footnote{In practice, lost customers includes those who abandoned, disconnected, or were blocked, which we collapse into one phenomenon in the model.}
For abandonment, the case data and documents show that not only are abandoning callers quite likely to call back, there are different cadences to when those re-dials occur. Some abandoning callers call back on the order of hours later, whereas others call back a number of days later. We deem these short and long re-dial orbits, respectively. These re-dial orbits highlight the \textit{negative feedback loops} in congested social services. Then, for callers who are able to get connected to an agent and complete their call, we view two possible paths: either the call was successful, and the caller becomes enrolled in benefits, or the call was not successful, and the caller will eventually need to call back. For the latter, we direct unsuccessful service completions into the long re-dial orbit, and, for the former, we direct to an even longer returning orbit, on the order of months, to represent the eventual need to re-certify eligibility for benefits. To that end, we model two possible true exits from the system. First, on the enrolled orbit, we allow some probability that the recipients will not pursue their next re-certification; such attrition can arise from exogenous events that remove households from eligibility. Second, we suppose that there is some fraction of abandoning callers who truly do not ever call back, do not enter any orbit, and are thus lost from the system -- these are the model's procedural denials. 
\rev{\paragraph{Justifying the model design}
Before we continue the technical exposition of the model, let us first highlight evidentiary sources for our key model additions, namely abandonments converting into long and short redials. Significant abandonment is well-documented in topline metrics, in our data (54.5\% on average over the call centers, weighted by call volume - see \Cref{tab:call_center_stats}), and in other social service call centers \citep{warren}. Although the data doesn't track redials, key case facts turn on redials as evidence of system failures: Plaintiff Holmes dialed 3 times in a week, Plaintiff Dallas dialed 10 times in 3 days, and Plaintiff David called multiple times in a day and week \citep{holmesKnodellPlaintiffsMSJ2024}. The redial \textit{timing} is informed by case facts, policy context, as well as a time series analysis of call center volume and its lagged dependence (see \Cref{apx-callcenterarrival-ts} for more details). \textit{Long redials} arise from caseworker requests for additional documentation, such as employment/income verification or paystubs \citep{seim2026welfare}.\footnote{At least in California, applicants have at least 10 days to provide such information \citep{BenefitsCal}.} We verify this empirically via a time series analysis, which finds statistically significant temporal dependence both on $1-2$ and $7-9$ days prior. In summary, our key model choices triangulate across different sources including various policy documents and federal guidance \citep{fns_snap_overview}, court documents \citep{holmesKnodellPlaintiffsMSJ2024}, ethnography \citep{seim2026welfare}, and user experience/service design literature \citep{pahlka2023recoding} to isolate the key differentiators from prior queueing guidance.  }

\begin{figure}[htbp]
\centering
\begin{tikzpicture}[>=latex, rotate=90]

    \draw (-0.1,-0.1) -- ++(0,0.5cm) -- ++(3.6cm,0)  -- ++(0,-0.5cm);
\node[align=center] at (1.35cm,0.925cm) {\scriptsize \textcolor{blue}{Call Center}};
\node[align=center] at (1.7cm,0.85cm) {$Q(t)$};

\draw (0,0) -- ++(2cm,0) -- ++(0,-1.5cm) -- ++(-2cm,0);
\foreach \i in {1,...,4}
  \draw (2cm-\i*10pt,0) -- +(0,-1.5cm);

\draw (2.75,-0.75cm) circle [radius=0.75cm];

\draw[<-, thick] (0,-0.75) -- +(-50pt,0) node[above left] {$\lambda$};
\node[align=center] at (-1.55cm,-1.1cm) {\scriptsize \textcolor{magenta}{Fresh}};

\node at (2.75,-0.75cm) {$c$};

\draw (0.5,2.5) -- ++(2cm,0) -- ++(0,-1cm) -- ++(-2cm,0) --++ (0,1cm);
\node[align=center] at (1.5cm,2cm) {$B(t)$};
\node[align=center] at (1.15cm,2cm) {\scriptsize \textcolor{blue}{Enrolled}};
\draw[->, thick] (1.5cm,2.5cm) -- +(0,35pt) node[above right] {\small $\gamma B$};
\node[align=center] at (1.2cm,3.25cm) {\scriptsize \textcolor{magenta}{Depart}};

\draw[->, thick] (3.5,-0.75) .. controls (5,0) and (4,1.75) .. (2.5,2);
\node[align=center] at (4.45cm,0.5cm) {\small $\mu_{\mathsf{+}} (Q \wedge c)$};
\node[align=center] at (4cm,0.35cm) {\scriptsize \textcolor{magenta}{Re-enroll}};

\node[align=center] at (2.175cm,-2.25cm) {\scriptsize \textcolor{magenta}{Abandon}};
\draw[->, thick] (1,-1.5) .. controls (4,-2) and (3.25,-3.5) .. (2.5,-3.5);
\node[align=center] at (2.675cm,-2.8cm) {\small $\theta_\mathsf{S} (Q - c)^+$};
\draw[->, thick] (1,-1.5) .. controls (4.75,-2) and (4.75,-4.5) .. (2.5,-4.95);
\node[align=center] at (3.7cm,-3.525cm) {\small $\theta_\mathsf{L} (Q - c)^+$};
\draw[->, thick] (1,-1.5) .. controls (4.25,-2) and (5,-3) .. (5.5,-3.5);
\node[align=center] at (5.3cm,-2.6cm) {\small $\theta_\mathsf{A} (Q - c)^+$};
\node[align=center] at (5.35cm,-3.7cm) {\scriptsize \textcolor{magenta}{Lost}};

\draw[->, thick] (3.5,-0.75) .. controls (5.5,-2.5) and (4.25,-5) .. (2.5,-5.05);
\node[align=center] at (4.55cm,-4.4cm) {\small $\mu_{\mathsf{-}} (Q \wedge c)$};
\node[align=center] at (4.25cm,-4.8cm) {\scriptsize \textcolor{magenta}{Re-connect}};

\draw[->, thick] (0.5,-3.5) .. controls (-2,-3) and (-1,-1.25) .. (0,-1.05);
\node[align=center] at (-0.35cm,-2.85cm) {\small $\delta_\mathsf{S} R_\mathsf{S}$};
\draw[->, thick] (0.5,-5) .. controls (-2.5,-4.5) and (-1.5,-1.25) .. (0,-0.9);
\node[align=center] at (-0.55cm,-4.25cm) {\small $\delta_\mathsf{L} R_\mathsf{L}$};
\node[align=center] at (-0.8cm,-3.5cm) {\scriptsize \textcolor{magenta}{Re-dial}};

\draw[->, thick] (0.5,2) .. controls (-1.75,1.5) and (-1,-0.5) .. (0,-0.6);
\node[align=center] at (-1.1cm,1.4cm) {\small $\delta_\mathsf{B} B$};
\node[align=center] at (-1.4cm,1.2cm) {\scriptsize \textcolor{magenta}{Re-certify}};

\draw (0.5,-3) -- ++(2cm,0) -- ++(0,-1cm) -- ++(-2cm,0) --++ (0,1cm);
\node[align=center] at (1.5cm,-3.5cm) {$R_\mathsf{S}(t)$};
\node[align=center] at (1.2cm,-3.5cm) {\scriptsize \textcolor{blue}{Short Orbit}};

\draw (0.5,-4.5) -- ++(2cm,0) -- ++(0,-1cm) -- ++(-2cm,0) --++ (0,1cm);
\node[align=center] at (1.5cm,-5cm) {$R_\mathsf{L}(t)$};
\node[align=center] at (1.2cm,-5cm) {\scriptsize \textcolor{blue}{Long Orbit}};

\end{tikzpicture}
\vspace{-.3in}
\caption{Process flow diagram of a queueing theoretic model of the SNAP call center.}\label{queueFig}
\end{figure}
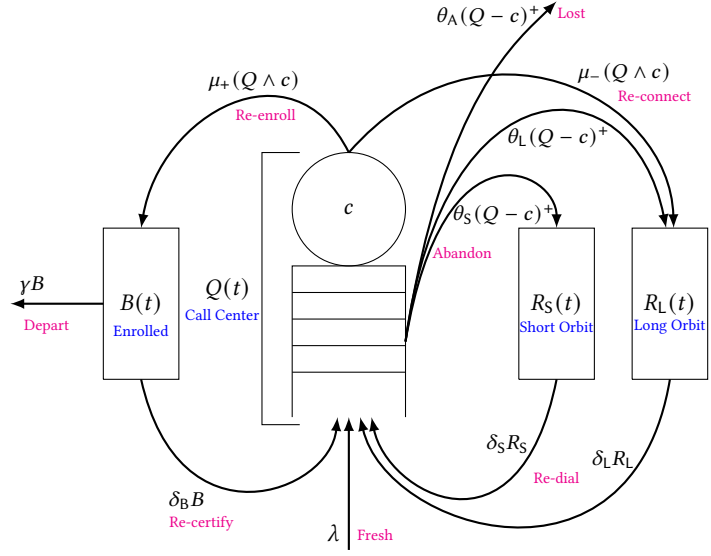

\paragraph{Definition of the SNAP call center stochastic model}

In a manner closely following \citet{ding2015fluid} but with further details added for the sake of the problem context, let us model the call center as a Markovian service network.  First, let $Q(t)$ be the number of callers either actively speaking to one of the $c \in \mathbb{Z}_+$ agents or waiting for the next available agent at time $t \geq 0$. Let $\lambda > 0$ be the rate of ``fresh'' or exogenous arrivals to the call center for those newly seeking to enroll in benefits. Suppose that each agent completes calls at rate $\mu = \mu_{\mathsf{+}} + \mu_{\mathsf{-}}$ for $\mu_{\mathsf{+}}, \mu_{\mathsf{-}} > 0$, where each completed call has a probability $\mu_{\mathsf{+}}/\mu$ of being {successfully} completed, meaning that the caller will not need to call back in order to complete enrollment and receive benefits, and probability $\mu_{\mathsf{-}}/\mu$ of being unsuccessful, meaning that the caller will need to call back. Let $B(t)$ track the number of successfully enrolled recipients at time $t$. Successfully enrolled recipients will either eventually ``qualify out'' of the system or otherwise need to re-certify. 
At rate $\delta_\mathsf{B} B(t)$, $B(t)$ will decrease by one with $Q(t)$ simultaneously increasing by one, and at rate $\gamma B(t)$, $B(t)$ will simply be decremented. 

Now, we will suppose that callers waiting for the next available agent each have independently and exponentially patience time with rate $\theta = \theta_\mathsf{A} + \theta_\mathsf{S} + \theta_\mathsf{L} > 0$ for $\theta_\mathsf{A}, \theta_\mathsf{S}, \theta_\mathsf{L} \geq 0$. With probability $\theta_\mathsf{A}/\theta$, the caller will abandon the queue and never call back to re-attempt enrollment. Then, with probability $\theta_\mathsf{S}/\theta$, the caller will hang up but eventually ``re-dial'' or attempt to call back after a (likely) short amount of time. Let $R_\mathsf{S}(t)$ be the number of callers at time $t$ who have left the queue but will soon re-dial, where each re-dial time is independent and exponentially distributed with rate $\delta_\mathsf{S}$. Finally, with probability $\theta_\mathsf{L}/\theta$, a caller will hang up and eventually call back but after a long time like the unsuccessful callers; hence, we count both types of callers among $R_\mathsf{L}(t)$. On aggregate, at rate $\delta_\mathsf{S} R_\mathsf{S}(t)$, $R_\mathsf{S}(t)$ will decrease by one with $Q(t)$ increasing by one, and, likewise, at rate $\delta_\mathsf{L} R_\mathsf{L}(t)$, one caller from $R_\mathsf{L}(t)$ will transition back to $Q(t)$.

In this way, the full model $(Q, B, R_\mathsf{S}, R_\mathsf{L})$ is a multi-station queueing network where the $Q$ station is a $\cdot/M/c$ queue and each of $B$, $R_\mathsf{S}$, and $R_\mathsf{L}$ are $\cdot/M/\infty$ queues. We visualize this queueing network in Figure~\ref{queueFig}. 
Formally, the stochastic process $(Q,B,R_\mathsf{S}, R_\mathsf{L})$ can be constructed through Poisson processes:
\begin{align}
Q(t)
\label{QDef}
&=
Q(0)
+
\Pi_\lambda\left(t\right)
+
\Pi_{\mathsf{B}}\left(\textstyle{\int_0^t} \delta_\mathsf{B} B(s) \mathrm{d}s\right)
+
\Pi_{\mathsf{S}}\left(\textstyle{\int_0^t} \delta_\mathsf{S} R_\mathsf{S}(s) \mathrm{d}s\right)
+
\Pi_{\mathsf{L}}\left(\textstyle{\int_0^t} \delta_\mathsf{L} R_\mathsf{L}(s) \mathrm{d}s \right)
\\
&
\quad
-
\Pi_{\mathsf{Q:A}}\left(\textstyle{\int_0^t} \theta_\mathsf{A} \left(Q(s) - c \right)^+ \mathrm{d}s \right)
-
\Pi_{\mathsf{Q:S}}\left(\textstyle{\int_0^t} \theta_\mathsf{S} \left(Q(s) - c \right)^+ \mathrm{d}s \right)
-
\Pi_{\mathsf{Q:L}}\left(\textstyle{\int_0^t} \theta_\mathsf{L} \left(Q(s) - c \right)^+ \mathrm{d}s \right)
\nonumber
\\
&
\quad
-
\Pi_{\mathsf{Q:+}}\left(\textstyle{\int_0^t} \mu_\mathsf{+}\left(Q(s) \wedge c \right) \mathrm{d}s\right)
-
\Pi_{\mathsf{Q:-}}\left(\textstyle{\int_0^t} \mu_\mathsf{-}\left(Q(s) \wedge c \right) \mathrm{d}s\right)
,
\nonumber
\\
B(t)
\label{BDef}
&=
B(0)
+
\Pi_{\mathsf{Q:+}}\left(\textstyle{\int_0^t} \mu_\mathsf{+}\left(Q(s) \wedge c \right) \mathrm{d}s\right)
-
\Pi_{\mathsf{B}}\left(\textstyle{\int_0^t} \delta_\mathsf{B} B(s) \mathrm{d}s\right)
,
\\
R_\mathsf{S}(t)
\label{RSDef}
&=
R_\mathsf{S}(0)
+
\Pi_{\mathsf{Q:S}}\left(\textstyle{\int_0^t} \theta_\mathsf{S} \left(Q(s) - c \right)^+ \mathrm{d}s \right)
-
\Pi_{\mathsf{S}}\left(\textstyle{\int_0^t} \delta_\mathsf{S} R_\mathsf{S}(s) \mathrm{d}s\right)
,
\\
R_\mathsf{L}(t)
\label{RLDef}
&=
R_\mathsf{L}(0)
+
\Pi_{\mathsf{Q:-}}\left(\textstyle{\int_0^t} \mu_\mathsf{-}\left(Q(s) \wedge c \right) \mathrm{d}s\right)
+
\Pi_{\mathsf{Q:L}}\left(\textstyle{\int_0^t} \theta_\mathsf{L} \left(Q(s) - c \right)^+ \mathrm{d}s \right)
-
\Pi_{\mathsf{L}}\left(\textstyle{\int_0^t} \delta_\mathsf{L} R_\mathsf{L}(s) \mathrm{d}s \right)
.
\end{align}
Here, as in the literature \citep[e.g.,][]{mandelbaum1998strong,pang2007martingale}, we use $\Pi_i(\cdot)$ for $i \in \mathcal{I} = \{\lambda, \mathsf{B}, \mathsf{S}, \mathsf{L}, \mathsf{Q:A}, \mathsf{Q:S}, \mathsf{Q:L}, \mathsf{Q:+}, \mathsf{Q:-}\}$ are mutually independent unit-rate Poisson processes. We assume that $Q(0)$, $B(0)$, $R_\mathsf{S}(0)$, and $R_\mathsf{L}(0)$ are known initial conditions.

\paragraph{Model interpretability via Poisson thinning and superposition} 
Because the queueing network model is composed from a collection of Poisson processes, we can leverage well-known probabilistic properties to further interpret its transitions.
The transition rates can equivalently be interpreted through standard Poisson thinning: abandoning callers enter the lost, short-redial, and long-redial paths in proportions $\theta_A/\theta$, $\theta_S/\theta$, and $\theta_L/\theta$, while completed calls succeed with probability $\mu_+/\mu$ and otherwise enter the reconnect orbit. See \Cref{apx-addl-discussion} for more explanation.

\subsection{Fluid model}

As a simpler companion to the stochastic queueing model, let us also define the deterministic \emph{fluid model} $(q,b,r_\mathsf{S}, r_\mathsf{L})$ through the following system of ordinary differential equations (ODEs):
\begin{align}
&
\dot{q}(t)
\label{qDef}
=
\lambda
+
\delta_\mathsf{B} b(t)
+
\delta_\mathsf{S} r_\mathsf{S}(t)
+
\delta_\mathsf{L} r_\mathsf{L}(t)
-
\theta (q(t) - c)^+
-
\mu (q(t) \wedge c)
,
&&
\dot{b}(t)
=
\mu_\mathsf{+} (q(t) \wedge c)
-
(\delta_\mathsf{B} + \gamma) b(t)
,
\\
&
\dot{r}_\mathsf{L}(t)
\label{rlDef}
=
\mu_\mathsf{-} (q(t) \wedge c)
+
\theta_\mathsf{L} (q(t) - c)^+
-
\delta_\mathsf{L} r_\mathsf{L}(t)
,
&&
\dot{r}_\mathsf{S}(t)
=
\theta_\mathsf{S} (q(t) - c)^+
-
\delta_\mathsf{S} r_\mathsf{S}(t)
.
\end{align}
As one may recognize by comparison to Figure~\ref{queueFig}, the dynamics of these ODEs align with the overall flow of the queueing model. However, the fluid model is deterministic, unlike the stochastic queueing model, suggesting that the fluid model could be a good first-order approximation for the queue. There is a clean intuition for that approximation: Imagine the call center queue as a physical system, where callers actually arrive to process through the network as visualized in Figure~\ref{queueFig}. Now, imagine those customers as arriving faster and faster, but also proportionally shrinking smaller and smaller, so that they move through the various services and orbits faster and faster as well. As this shrinking and speeding goes to the extreme, the queueing network would look like a continuous flow of callers in and out of the stations of the queue, where the random times between arrivals, abandonments, completed calls, re-dials, re-connects, and re-enrolls would be replaced by fluid at the same rates. Moreover, following intuition granted by the law of large numbers, the stochastic fluctuations brought by those random arrival, abandonment, and service times should eventually be dominated by the sheer scale of the growing and accelerating system.

This intuition can be formalized: the fluid model indeed arises out of a functional strong law of large numbers (FSLLN) limit of the queueing model, and we now formalize this connection in Proposition~\ref{fluidLimit}.

\begin{proposition}\label{fluidLimit}
For $n \geq 0$, let $\left(Q^{(n)}, B^{(n)}, R_\mathsf{S}^{(n)}, R_\mathsf{L}^{(n)}\right)$ be the queueing network model of~\eqref{QDef} through~\eqref{RLDef} alternately defined with external arrival rate $\lambda n$ and number of servers $c n$. Then, given that
\begin{align}
\lim_{n \to \infty}
\left(\frac{1}{n}Q^{(n)}(0), \frac{1}{n}B^{(n)}(0), \frac{1}{n}R_\mathsf{S}^{(n)}(0), \frac{1}{n}R_\mathsf{L}^{(n)}(0)\right)
\longrightarrow
\left(q^0, b^0, r_\mathsf{S}^0, r_\mathsf{L}^0\right)
,
\end{align}
the fluid limit of the scaled stochastic process converges to the deterministic system,
\begin{align}
\lim_{n \to \infty}
\left(\frac{1}{n}Q^{(n)}(t), \frac{1}{n}B^{(n)}(t), \frac{1}{n}R_\mathsf{S}^{(n)}(t), \frac{1}{n}R_\mathsf{L}^{(n)}(t)\right)
\stackrel{\mathsf{a.s.}}{\longrightarrow}
\left(q(t), b(t), r_\mathsf{S}(t), r_\mathsf{L}(t)\right)
,
\end{align}
uniformly on compact sets, where $\left(q(t), b(t), r_\mathsf{S}(t), r_\mathsf{L}(t)\right)$ is the unique solution to the system of equations in~\eqref{qDef} through~\eqref{rlDef} with initial condition $\left(q(0), b(0), r_\mathsf{S}(0), r_\mathsf{L}(0)\right) = \left(q^0, b^0, r_\mathsf{S}^0, r_\mathsf{L}^0\right)$.
\end{proposition}
\begin{proof}
This fluid limit can be readily obtained from prior functional strong law results from the literature; for instance, the proof of this limit follows immediately from theorem 2.2 of \citet{mandelbaum1998strong}.
result. \rev{Alternatively, one can also obtains this fluid limit via the standard ``recipe'' of \citet{ethier2009markov}, as \citet{ding2015fluid} does excellently. 
}
\end{proof}

From this combination of intuition and rigorous connection, fluid models are very commonly used to study queueing models in a simplified manner that grants tractability but still captures the heart of the modeling context.\footnote{Critically, let us note that the solutions to the fluid model equations need not exactly match the underlying means of the stochastic model, especially for queues with abandonment \citep[e.g.,][]{daw2019new}. Nevertheless, for large systems with many arrivals and many servers (as in our motivating context), the stochastic model is typically close enough to the fluid approximation for its first-order insights to be both valid and valuable, as implied by Proposition~\ref{fluidLimit}.} 
These benefits are particularly evident when studying the steady state of the system, in which the system of ODEs simplify to a system of linear equations.
By setting the equations in~\eqref{qDef} and~\eqref{rlDef} each to 0 and solving for an equilibrium solution, which we denote $(\bar q, \bar b, \bar r_\mathsf{S}, \bar r_\mathsf{L})$, this fluid model admits a first-order approximation for the system's steady-state means:
\begin{align}
\bar q
&=
\left(
c \wedge \left(1 + \frac{\delta_\mathsf{B}}{\gamma}\right)\frac{\lambda}{\mu_\mathsf{+}}
\right)
+
\frac{1}{\theta_\mathsf{A}}
\left(
\lambda - \frac{\gamma c \mu_{\mathsf{+}}}{\gamma + \delta_\mathsf{B}} 
\right)^+
,
\qquad
&&
\bar b 
=
\left(
\frac{\lambda}{\gamma}
\wedge
\frac{c\mu_{\mathsf{+}}}{\gamma + \delta_\mathsf{B}}
\right)
,
\\
\bar r_\mathsf{L}
&=
\left(
\frac{c\mu_{\mathsf{-}}}{\delta_\mathsf{L}}
\wedge
\left(
1 + \frac{\delta_\mathsf{B}}{\gamma}
\right)
\frac{\lambda \mu_{\mathsf{-}}}{\delta_\mathsf{L}\mu_{\mathsf{+}}}
\right)
+
\frac{\theta_\mathsf{L}}{\delta_\mathsf{L} \theta_\mathsf{A}}
\left(
\lambda - \frac{\gamma c \mu_{\mathsf{+}}}{\gamma + \delta_\mathsf{B}} 
\right)^+
,
\qquad
&&
\bar r_\mathsf{S}
=
\frac{\theta_\mathsf{S}}{\delta_\mathsf{S} \theta_\mathsf{A}}
\left(
\lambda - \frac{\gamma c \mu_{\mathsf{+}}}{\gamma + \delta_\mathsf{B}} 
\right)^+
.
\end{align}

\paragraph{Distinguishing arrival rates in the fluid model} By comparison to the \emph{fresh} arrival rate, $\lambda$, which only counts callers that are newly seeking benefits, the overall rate of arrivals to the system will also include re-dialing callers and current benefit recipients seeking to re-enroll. Because the total arrival rate is both an important quantity in its own right and more readily observed in data, let us formalize it in notation. Specifically, in the fluid model, let 
$\hat \lambda(t) 
= 
\lambda
+
\delta_\mathsf{B} b(t)
+
\delta_\mathsf{S} r_\mathsf{S}(t)
+
\delta_\mathsf{L} r_\mathsf{L}(t)$ 
denote the \rev{\emph{total arrival rate}} to the system at time $t$, and, without the $t$ argument, let 
$\hat \lambda
= 
\lambda
+
\delta_\mathsf{B} \bar{b}
+
\delta_\mathsf{S} \bar{r}_\mathsf{S}
+
\delta_\mathsf{L} \bar{r}_\mathsf{L}$ 
be the total arrival rate for the fluid model in steady-state. 

Furthermore, in studying the caller experience, it will also be of interest to focus on the volume of arrivals that includes the re-dialing and re-enrolling orbits but excludes the customers that renege from the queue before connecting to an agent. To be able to compute values such as the average waiting time among callers that do not abandon, let us define $\tilde{\lambda}(t) = \hat\lambda(t) - \theta(q(t) - c)^+$ as the \rev{\emph{effective arrival rate}, i.e. the total arrival rate at time $t$ less abandonments}, and we again let the argument-less $\tilde \lambda$ denote \rev{the effective arrival rate in steady-state}.

\paragraph{Simplifications in the overloaded regime.} When the system is overloaded, i.e. $q(t) \geq c$, then we can simplify the branching logic so that $(q(t) - c)^+ = q(t) -c$ and $(q(t) \wedge c)=c$. Moreover, in  steady-state, the overloaded regime can be easily characterized, and it yields simple expressions for the four model quantities.



\begin{corollary}\label{overCor}
The overloaded regime, $\bar{q} \geq c$, occurs as the equilibrium fluid model solution if and only if $\lambda \geq {\gamma c \mu_{\mathsf{+}}}/({\gamma + \delta_\mathsf{B}})$, which equivalently occurs if and only if $\hat{\lambda} \geq c\mu$. In this case, the steady-state solutions simplify to
\begin{align}
\bar{q}
&= \frac{1}{\theta_{\mathsf{A}}}
\left(\lambda - \frac{\gamma c\mu_{\mathsf{+}} }{\gamma + \delta_{\mathsf{B}}}\right) 
+ c, \qquad 
\bar{b}
= \frac{c\,\mu_{\mathsf{+}}}{\gamma + \delta_{\mathsf{B}}}, 
\qquad 
\bar{r}_S = \frac{\theta_{\mathsf{S}}}{\theta_{\mathsf{A}} \delta_{\mathsf{S}}}
\left(\lambda - \frac{\gamma c\mu_{\mathsf{+}} }{\gamma + \delta_{\mathsf{B}}}\right) ,
\qquad 
\bar{r}_L= \frac{c\mu_{\mathsf{-}}}{\delta_{\mathsf{L}}}
+ \frac{\theta_{\mathsf{L}}}{\delta_{\mathsf{L}} \theta_{\mathsf{A}}}
\left(\lambda - \frac{\gamma c\mu_{\mathsf{+}} }{\gamma + \delta_{\mathsf{B}}}\right) .\label{eqn-steady-state-simplified}
\end{align}
\end{corollary}


In the remainder of the paper, we will largely focus on the steady-state fluid model and assume it to be in the overloaded regime. This parsimonious approximation offers the simplest first step to derive insights from the \emph{Holmes v. Knodell} data and assess the impacts of possible changes to the SNAP call center system. 
The overloaded assumption is well-justified, here and more broadly, since system-level stresses motivate us to study the problem in the first place.
 
\section{Data and Model Fitting}

In this section, we illustrate how to connect our model parameters to call center data in general, and we discuss how we do so for Missouri DSS' call centers via the aggregated data released in court documents. 


\textbf{Fitting model parameters in general: }

\textit{Simple parameter matching from observable metrics: }The eleven parameters used by the queueing model
may not actually all be readily identifiable in practice. For instance, the model uses $\lambda$, the rate of \emph{fresh} arrivals only, as a parameter, whereas $\hat \lambda$, the total rate of all arrivals, is likely what is recorded operationally either by hand or by typical call center management software. Rather, disentangling $\lambda$ from $\hat \lambda$ is actually a possible benefit of the model. 

Let us consider what quantities should be observable in practice and connect them to the (overloaded) steady-state fluid approximation.\footnote{If fine-grained data on arrival times and call durations is available, more rigorous techniques like maximum likelihood estimation can deliver high fidelity parameter estimation. But since the fluid model is itself already a deterministic approximation to the queueing model's stochastic approximation of reality, these back-of-the-envelope calculations can be used to guide calibration of the fluid model from aggregated data.} 
Several inputs are directly observable or administratively specified: total arrivals $\hat{\lambda}$, staffing $c$, average handling time $1/\mu$\footnote{This service time \textit{excludes} waiting and its calculation \textit{excludes} any callers who abandon before being connected to an agent.}, the re-certification interval $1/\delta_\mathsf{B}$, $\bar{q}$ and $\bar{b}$, the average number of present callers (including both waiting and in service), and benefits retention, which pins down $\gamma$ through $\delta_\mathsf{B}/(\gamma+\delta_\mathsf{B})$. 

Assuming that the system is in the overloaded regime, we can take the above assumptions of observable quantities within the equilibrium fluid model and the definitions of $\hat \lambda$, $\theta$, and $\mu$ to provide a system of equations for the remaining model quantities. We do so now, and we underline the unknown terms in each equation for emphasis:
\begin{align*}
\begin{alignedat}{4}
0&=\underline{\lambda}+\delta_\mathsf{B}\bar b
+\underline{\delta_\mathsf{S}}\underline{\bar r_\mathsf{S}}
+\underline{\delta_\mathsf{L}}\underline{\bar r_\mathsf{L}}
-\theta(\bar q-c)-c\mu,
&\qquad
0&=c\underline{\mu_+}-(\gamma+\delta_\mathsf{B})\bar b,
&\qquad
0&=\underline{\theta_\mathsf{S}}(\bar q-c)
-\underline{\delta_\mathsf{S}}\underline{\bar r_\mathsf{S}},
&\qquad
0&=c\underline{\mu_-}
+\underline{\theta_\mathsf{L}}(\bar q-c)
-\underline{\delta_\mathsf{L}}\underline{\bar r_\mathsf{L}},
\\
\hat\lambda
&=\underline{\lambda}+\delta_\mathsf{B}\bar b
+\underline{\delta_\mathsf{S}}\underline{\bar r_\mathsf{S}}
+\underline{\delta_\mathsf{L}}\underline{\bar r_\mathsf{L}},
&\qquad
\theta&=\underline{\theta_\mathsf{A}}
+\underline{\theta_\mathsf{L}}
+\underline{\theta_\mathsf{S}},
&\qquad
\mu&=\underline{\mu_-}+\underline{\mu_+}.
\end{alignedat}
\end{align*}
As written, this system has seven equations and ten unknowns. Hence, we need more information to solve for the missing values, albeit perhaps not as much as it might seem. First, consider $\delta_\mathsf{S}$, $\delta_\mathsf{L}$, $\bar{r}_\mathsf{S}$, and $\bar{r}_\mathsf{L}$. Several approaches can quickly reduce these four to just two missing values. For instance, one can estimate $\delta_\mathsf{S}$ and $\delta_\mathsf{L}$ from time-series data, or work directly with the \emph{volume} of returning callers per unit of time, say $\bar{v}_\mathsf{S} = \delta_\mathsf{S} \bar{r}_\mathsf{S}$ and $\bar{v}_\mathsf{L} = \delta_\mathsf{L} \bar{r}_\mathsf{L}$, which preserves our ability to evaluate key performance metrics. 
Taking such reduction for granted, we are left with seven linear equations and eight unknowns. We can 
pin down this final degree of freedom 
in one of three ways: use domain knowledge to fix the short- versus long-redial split, simplify to a single redial orbit ($\theta_\mathsf{S}=0$), or sweep over plausible abandonment/redial parameters and compare the implied metrics to auxiliary data. We take the last approach in our analysis.

\textbf{Fitting Parameters to the Holmes vs. Knodell Data:}

\rev{\textit{Dataset description: }
We downloaded the Holmes vs. Knodell court documents from the PACER database, which include extensive reports (typically monthly, weekly or daily). 
We primarily analyze data from ``Exhibit 87'''s daily aggregated reports from four different call centers, varying in size, spanning 9/29/2021-12/29/2023, though not all call centers report information for the whole time period. (Some of these call centers are centralized CSCs that handle calls for other DSS programs such as Medicaid.) 
We observe information about: call volume (number of incoming calls, calls answered successfully, calls abandoned); wait time in queue (before successful connection or abandonment, ASA [average speed to answer] before successful connection, AHT [average handling time] spent by agent on a call's duration); and staffing.
Average daily call volume per call center is around $(700, 2500, 3100, 6800)$ with average abandonment fractions $(2.5\%, 58\%, 57\%, 57\%)$. The large call volume justifies our fluid approximation. Given the heterogeneity in call center volumes and staffing, we later fit parameters separately to each call center. See \Cref{apx-callcenterdata} for more details.}

\textit{Parametrizing the Queueing Model from Aggregated Call Center Data: }We now discuss how we fit parameters specifically for the aggregated call center and arrival data for court documents from \textit{Holmes v. Knodell}, see \Cref{apx-callcenterdata} for more detail. 
We parametrize everything in terms of the number of minutes a call center is open in a day, $\rev{MMin}=540 \text{min}=9 \text{ hours}$; the call centers are open on weekdays. In our model, some \textit{operational} parameters such as the total arrival rate $\hat{\lambda}$, staffing levels $c$, and average handling time (AHT) are directly observed in the call-center data. 
To reduce parameter dimensionality, we fix a proportion $p_+\in[0,1]$ of completed calls that result in completed re/-enrollments, hence $\mu_{\mathsf{+}} = p_+ \left( \frac{\rev{MMin}}{AHT} \right).$ 
We set $\delta_{\mathsf{B}},\delta_{\mathsf{S}},\delta_{\mathsf{L}},p_+$ by using a combination of program design, domain knowledge, and time series analysis. \rev{\textit{Program design} determines $\delta_{\mathsf{B}} = 1/128.5$ since per SNAP eligibility guidelines, the recertification period is 6 calendar months $\approx$ 128.5 model (call center) days. 
\textit{Domain knowledge} sets $\delta_{\mathsf{S}}=3$, corresponding to a typical callback time of $3$ hours; this is a behavioral assumption about within-day callback behavior.} 
The remaining behavioral parameters governing abandonment and re-dial behavior, $(\lambda,\theta_{\mathsf{A}},\theta_{\mathsf{S}},\theta_{\mathsf{L}})$, namely the fresh arrival rate $\lambda$ and rates of total abandonment $\theta_A$, short/long orbit abandonment $\theta_S,\theta_L$, are not directly observed in our aggregated data.\footnote{From data, we only have measures of aggregated abandonment percentages (typically daily or weekly).}


\textit{Derived metrics and an operational dashboard: }The steady-state approximation allows us to derive \textit{model-based estimates} of performance metrics. The model results best inform \textit{relative comparisons} as to how changes in call center design \textit{increase or decrease key metrics}. We introduce the resulting \textit{model-based performance metrics}, which use the approximated steady state values and include: average wait time  ($\bar{w} = (\bar{q}-c)^+ / \hat \lambda$), which includes abandonments),  average speed to answer ($\tilde{w} = (\bar{q}-c)^+ /\tilde \lambda$), which does not include abandonments), abandonment fraction ($1 - \frac{\tilde{\lambda}}{\hat\lambda}$), total lost abandonment ($\theta_{\mathsf{A}} (\bar{q}-c)^+$), and endogenous congestion from re-dials ($ \delta_{\mathsf{S}} \bar r_{\mathsf{S}} + \delta_{\mathsf{L}} \bar r_{\mathsf{L}}$). 


We can interpret the metric of lost calls, $\theta_{\mathsf{A}} (\bar{q}-c)^+,$ those who completely leave the system, as a metric related to procedural denials. The fluid model is not sharp enough to estimate whether a single caller would get through in $30$ days, so \rev{our total abandonment metric is a policy-relevant lower bound on procedural denials. Of the other quantities, wait times directly affect customer experience and are key metrics in essentially any call center.}
\rev{Other relevant values include utilization $U = {\tilde{\lambda}}/{c \mu} ,$ which measures the ratio of the \emph{effective} arrival rate to service capacity and also captures the average fraction of service providers that are busy, and the 
implied utilization
    $IU  ={\hat\lambda}/{c\mu}$, which is a common adaptation of the utilization metric to \emph{all arrivals} that also measures overload (i.e., $IU$ can exceed 1 whereas $U \leq 1$).\footnote{\rev{By consequence of Corollary~\ref{overCor}, we can observe that, in the overloaded regime, $U=1$ and $IU \geq 1$.}}}

We develop an analytics dashboard where stakeholders could perturb system parameters to examine impacts on these performance metrics, see \Cref{fig:dashboard} in the Appendix for a screenshot.\footnote{The (anonymized version of) the dashboard is publicly available at \url{https://call-center-simulator-production.up.railway.app}.}

\textit{Robustly fitting $\lambda$ over a range of plausible $\theta_{\mathsf{A}}, \theta_{\mathsf{S}}, \theta_{\mathsf{L}}$: } 
Since the finer-grained behavioral parameters $\theta_{\mathsf{A}}, \theta_{\mathsf{S}}, \theta_{\mathsf{L}}$ are uninformed from data, we fit a single fresh arrival rate parameter $\lambda$ that recovers wait times well over a \textit{potential range of plausible values of $\theta_{\mathsf{A}}, \theta_{\mathsf{S}}, \theta_{\mathsf{L}}$}. We assess how the derived wait times $\bar{w},\tilde{w}$ compare to \textit{observed} wait times (including abandonment) and speed to answer (excluding abandonment). 
There are four different call centers in Exhibit 87. We group them by call volume and average handling time, two in a 20 minute AHT category and two in a 40 minute AHT category. We calculate the median average wait time and median average speed to answer. 
We sweep over a range of \textit{reasonable} $\theta_{\mathsf{A}},\theta_{\mathsf{S}},\theta_{\mathsf{L}}$; we fix the sum of rates $\theta_{\mathsf{A}}+\theta_{\mathsf{S}}+\theta_{\mathsf{L}}=10$ or $54$ (corresponding to $54$ or $10$ minute abandonment, respectively) and vary the relative weight of $\theta_{\mathsf{A}},\theta_{\mathsf{S}},\theta_{\mathsf{L}}$. Finally, we robustly fit $\lambda$ by minimizing the mean average deviation (averaging over $\theta$ parameter vector values) of median wait times.


  \begin{table}[t!]
  \centering
    \caption{Descriptive statistics on observed metrics for call centers (CC), and
  best-fit lambda values and MAD for each call center. Staffing and arriving
  calls are daily metrics.}\label{tab:parameter-fitting}

  \small
  \setlength{\tabcolsep}{4pt}
  \begin{tabular}{@{}cccccccc@{}}
  \toprule
  CC & c (staffing) & AHT (min) & Wait (min) & ASA (min) & observed $
  \hat{\lambda}$ & best-fit $\lambda$ & Best MAD ($\bar{w}$, $\tilde{w}$) \\
  \midrule
  1 & 32.5 & 20.8 & 4.4 & 4.1 & 228.0 & 145.4 & (1.2, 1.9) \\
  2 & 52.0 & 21.9 & 38.8 & 117.8 & 2931.5 & 623.3 & (2.9, 27.0) \\
  3 & 72.0 & 44.2 & 45.5 & 87.6 & 2396.5 & 700.1 & (11.6, 20.8) \\
  4 & 51.0 & 44.8 & 1.8 & 1.8 & 712.0 & 164.8 & (1.1, 1.6) \\
  \bottomrule
  \end{tabular}

  \vspace{-0.1in}
  \end{table}

\textit{Observed metrics and fitted parameters: }This table summarizes the key metrics for the four call centers, separating the empirically observed data from the fitted model parameters. The columns from ``c (staffing)'' through ``observed $\hat{\lambda}$'' represent the direct operational measurements. All time metrics are shown in decimal minutes. Specifically, ``Average Wait Time (min)'' and ``ASA (min)" correspond to the observed values for $\bar{w}$ and $\tilde{w}$, respectively, while ``observed $\hat{\lambda}$'' is the total arrival rate including re-dials. In contrast, the ``best-fit $\lambda$'' is a fitted parameter, representing the optimized robust rate of fresh arrivals derived from the model. The final column, ``Best MAD ($\bar{w}$, $\tilde{w}$)'', quantifies model fit with the Mean Absolute Deviation between model predictions and the ``Average Wait Time'' and ``ASA Time'', respectively.

\section{Evaluating system design changes}


\paragraph{Analytical results for operational changes in the general fluid model}

We begin by leveraging closed-form solutions for steady-state expressions for the performance metrics to analyze \textit{how} different potential operational changes affect system performance in the overloaded regime. Specifically, we now analyze six key steady-state performance metrics: the procedural denial rate ($\mathrm{PD} = \theta_A (\bar q - c)$), the mean number of waiting callers ($\bar{q} - c$), the mean waiting time ($\bar w = (\bar q - c) / \hat \lambda$), the average speed to answer ($\tilde w = (\bar q - c) / \tilde \lambda$), the endogenous congestion from re-dials ($\mathrm{EC}_\mathsf{R} = \delta_\mathsf{S} \bar r_\mathsf{S} + \delta_\mathsf{L} \bar r_\mathsf{L}$), and the endogenous congestion from re-certification ($\mathrm{EC}_\mathsf{B} = \delta_\mathsf{B} \bar b$).   For each quantity, we evaluate its dependence on four model parameters that could reasonably be subject to managerial control or operational design: the staffing level ($c$), the frequency of re-certification ($\delta_\mathsf{B}$), the individual service rate ($\mu$), and the fraction of successful calls (which we denote $p_+ = \mu_+ / \mu$, with $\mu$ assumed to be held fixed).\footnote{Though endogenous congestion from re-dials is undoubtedly bad for the system, the service agents, and the callers themselves, endogenous congestion from re-certifications might actually be considered at least somewhat good: it is a sign that people are successfully receiving SNAP benefits.} 

Intuitively, while $c$ may be expensive to increase, it may be relatively easy to re-design the eligibility interview process to be either faster (increase $\mu$) or more likely to result in certification (increase $p_+$), and, likewise, structural policy changes could increase the re-certification period (decrease $\delta_\mathsf{B}$). Here, we leverage the fluid model to offer insight of how these prospective changes may relate to one another.
In the interest of space, the exact expressions of the partial derivatives of each performance metric with respect to each operational parameter are left for the appendix (see Propositions~\ref{dPDprop} through~\ref{dECBprop}), and here we simply highlight their takeaways in a series of corollary results. 

First, we establish that any of the prospective changes will yield intuitive directional improvements in the procedural denial rate, the waiting time, the average speed to answer, and the number waiting. 






\begin{corollary}[Directional comparative statics]
In the overloaded regime, the procedural denial rate, mean waiting time, average speed to answer, and mean number waiting are each decreasing in $c$, $\mu$, and $p_+$ and increasing in $\delta_\mathsf{B}$. Endogenous congestion from re-certifications is increasing in each operational parameter. Endogenous congestion from re-dials is increasing in $\delta_\mathsf{B}$ and decreasing in $p_+$; its dependence on $c$ and $\mu$ is ambiguous, increasing with $c$ and $\mu$ if $\theta_\mathsf{A}(1-p_+)(\gamma+\delta_\mathsf{B})>p_+\gamma c(\theta_\mathsf{S}+\theta_\mathsf{L})$, and otherwise non-increasing in $c$ and $\mu$.
\end{corollary}

The main access-oriented metrics improve under candidate operational changes: increasing staffing, increasing service speed, increasing the fraction of successful calls, or lengthening the re-certification period. However, endogenous congestion from redials behaves differently: it decreases with better conversion or shorter re-certification periods, but, interestingly, it could either increase or decrease with higher staffing or faster service. 


To compare magnitudes rather than only signs, we next consider elasticities of the performance metrics with respect to the operational parameters. For a metric $y$ and parameter $x$, the local elasticity is $(x/y)\partial y/\partial x$. First, we find that changes to staffing level, service rate, and probability of success all yield the same respective elasticity for the procedural denial rate, number waiting, and endogenous congestion from re-certifications. For relative magnitudes, the key result is that staffing and service-speed improvements have identical proportional effects in the overloaded regime.

\begin{corollary}[Elasticity equivalences]
For each of PD, $\bar q-c$, and ECB, the elasticity is the same across $c$, $\mu$, and $p_+$. For each of $\bar w$, $\tilde w$, and ECR, the elasticity is the same across $c$ and $\mu$ (but different from the previous elasticity).
\end{corollary}



Then, for the waiting time, speed to answer, and endogenous congestion from re-dials, changes to the staffing level and to the service rate again yield equivalent elasticities, but they need not be the same as that of the success probability.
These corollaries together show that changes to $c$ and $\mu$ have the same respective impact for each of the six performance metrics, suggesting that they have equivalent operational impacts despite their likely quite different costs. 

\paragraph{Comparison of our endogenous-congestion-aware queueing model to generic Erlang-A staffing guidance}
In \Cref{tab:compact_staffing_benchmark}, we compare our model to staffing guidance from naive application of the standard Erlang-A model. Using the Erlang-A model to guide counterfactual staffing decisions is \textit{overly optimistic} and \emph{fundamentally under-staffed} since an abandoning customer leaves the system forever in the Erlang-A model, and therefore the phenomenon of abandonment actually lowers wait times overall. However, our proposed model with re-dials accurately represents the endogenous congestion where an abandoning caller today can actually increase congestion in the future when they re-dial or re-connect. Therefore, its requisite staffing prescriptions increase significantly beyond that of the Erlang-A model. In fact, in applying both models to the four call centers within the \emph{Holmes v. Knodell} data and computing staffing levels necessary for each to achieve standard performance metric targets, we find that the Erlang-A yields as large as a 84\% shortfall of the staffing prescribed from our endogenous-congestion-aware model's guidance. The largest discrepancies are driven by call centers where fitting Erlang-A to observed arrivals cannot fit observed wait times well, further demonstrating that Erlang-A's assumptions are unsuitable for redials in social service delivery.
  \begin{table}[htbp]
  \centering
  \caption{Our model vs. Erlang A. Entries
  in the target columns are required staffing counts to achieve that target. Shortfall is
  $(\mathrm{our\ model}-\mathrm{Erlang\ A})/\mathrm{our\ model}$. The
  abandonment target is the lost/abandoned share in our model $(1-\tilde \lambda /  \hat \lambda)$ and Erlang A's
  abandonment probability. The last column is Erlang A's absolute relative error
  in fitted average wait, $|\widehat{\bar w}_{EA}-\bar w|/\bar w$.}
  \label{tab:compact_staffing_benchmark}
  \begin{tabular}{lc ccc ccc c}
  \toprule
  && \multicolumn{3}{c}{Average wait $<1$ min} & \multicolumn{3}{c}{Abandonment
  $<10\%$} & \multicolumn{1}{c}{Erlang A $\bar w$ error} \\
  \cline{3-5}\cline{6-8}\cline{9-9}
  CC & $c$ & Our model & Erlang A & Shortfall & Our model &
  Erlang A & Shortfall & $|\Delta \bar w|/\bar w$ \\
  \hline
  1 & 32.5 & 52 & 13 & 75.0\% & 52 & 8 & 84.6\% & 100.0\% \\
  2 & 52 & 194 & 125 & 35.6\% & 174 & 107 & 38.5\% & 107.1\% \\
  3 & 72 & 217 & 203 & 6.5\% & 195 & 177 & 9.2\% & 22.1\% \\
  4 & 51 & 58 & 58 & 0.0\% & 58 & 56 & 3.4\% & 5.4\% \\
  \bottomrule
  \end{tabular}
\vspace{-0.1in}
  \end{table}

\paragraph{Possible improvements: staffing vs.~system design efficiency curves based on \emph{Holmes v. Knodell} data}

\begin{figure}[t!]
    \centering
    \subcaptionbox{Staffing-AHT curves, wait times; decreases in AHT\label{fig:wait_time-aht}}{%
\includegraphics[width=0.24\textwidth]{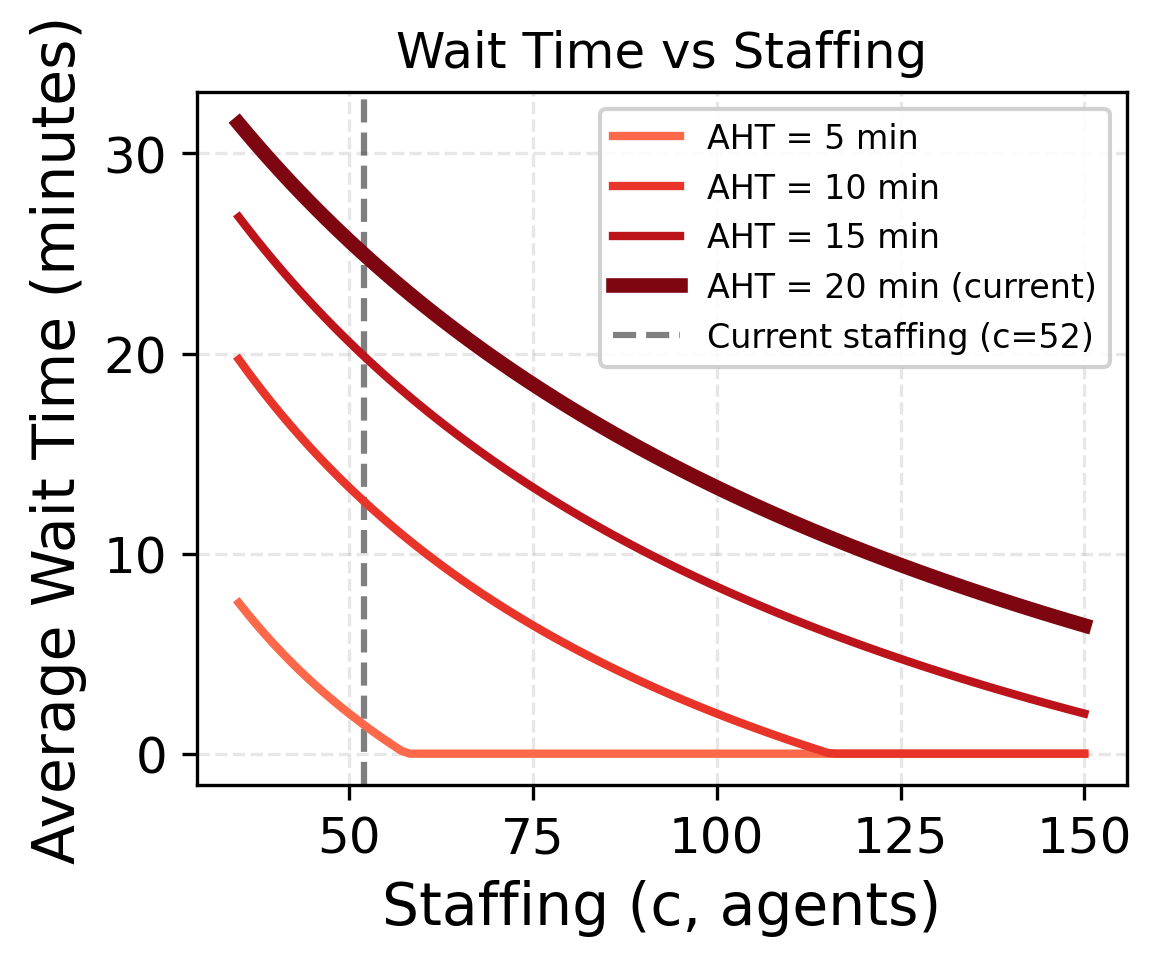}%
    }
    \hfill
    \subcaptionbox{Staffing-AHT curves, procedural denials\label{fig:denials-aht}}{%
        \includegraphics[width=0.24\textwidth]{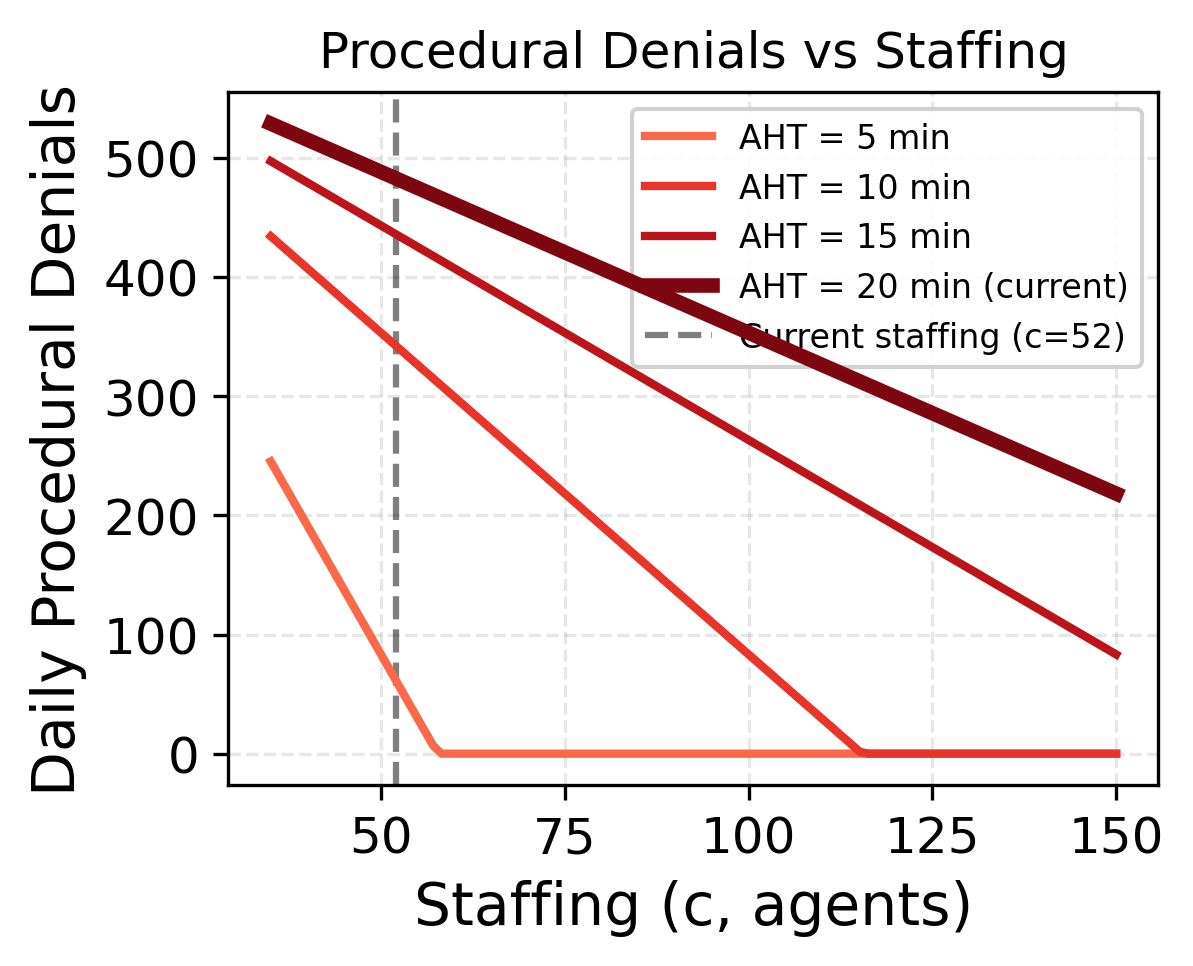}%
    }
        \subcaptionbox{Wait time vs. recertification length, $\delta_{\mathsf{B}}$\label{fig-wait-recert}}{%
\includegraphics[width=0.24\textwidth]{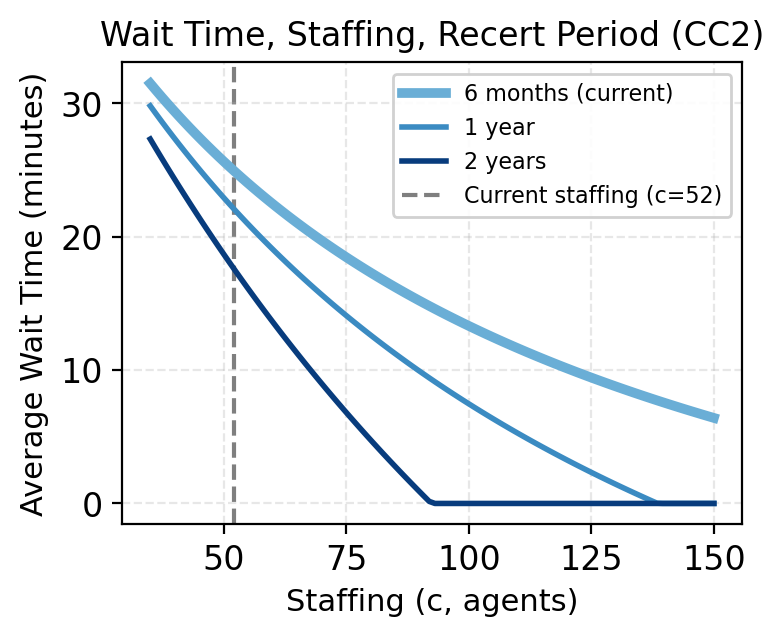}%
    }
    \hfill
    \subcaptionbox{Daily procedural denials vs recertification period (changes in $\delta_{\mathsf{B}}$). \label{fig:denials-recert}}{%
        \includegraphics[width=0.26\textwidth]{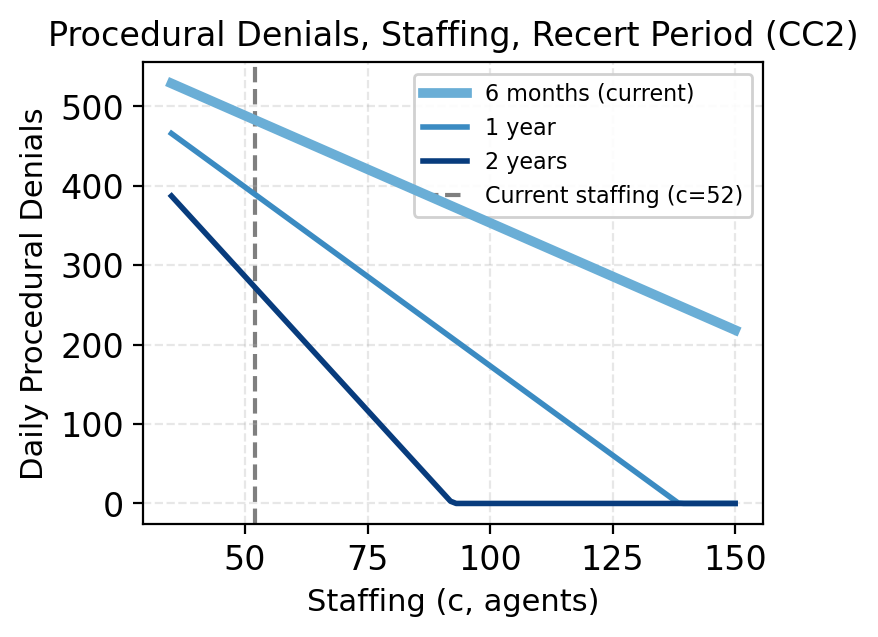}%
    }
    \vspace{-.2in}
    \caption{Performance metrics vs staffing levels for different potential changes in call center system design.}
    \label{fig:staffing-curves}
\end{figure}


While the prior analysis highlights the impact of potential system design changes one at a time, in \Cref{fig:staffing-curves} we leverage the model fitted to \textit{Holmes v. Knodell} call center data to explore the joint efficiency of \textit{simultaneous} improvements in staffing and other system design changes. Importantly, our richer model enables also assessing the impact of these other policy design and service delivery changes on performance metrics. \Cref{fig:staffing-curves} illustrates in \textit{AHT vs. staffing curves} how performance improvements in wait times and procedural denials depend on interactions between system design and staffing (on the x-axis) improvements. We generate these figures for one of the call centers (\#2); results for other call centers are qualitatively similar, hence omitted. We use a $\theta$ parameter vector $(\theta_{\mathsf{A}}, \theta_{\mathsf{S}},\theta_{\mathsf{L}}) =(4,3,3)$ --- although different values of $\theta_{\mathsf{A}}, \theta_{\mathsf{S}},\theta_{\mathsf{L}}$ result in different absolute magnitude estimates for wait times, all of our commentary depends on \textit{relative} orderings that remain the same, and procedural denial metrics depend only on $\theta_{\mathsf{A}}.$ \Cref{fig:wait_time-aht,fig:denials-aht} compare average wait-time and procedural denials on the y-axis as staffing changes, for different improvements in \textit{average handling time (AHT)}; each value of the AHT generates a different staffing curve.\footnote{System design changes that reduce AHT might include recent interest in AI chatbots, automated responses, increased training for staffing, or self-service options --- these could handle simple questions more briefly or redirect callers to new self-service channels.} In \Cref{fig:wait_time-aht} we see wait times decrease sharply at low staffing levels, with diminishing returns at higher staffing. Meanwhile, procedural denials (under our model --- total abandonment) decrease linearly in staffing but decreasing AHT increases the rate of enrolled individuals, and hence further improves the \textit{efficiency} of staffing improvements (slope of decrease in denials). Jointly considering system design improvements and staffing increases can \textit{increase the performance improvements} under practically limited additional staffing resources. Conversely, \textit{without} making changes in call center design and service process, it can take an unrealistically large number of additional staff to reduce wait times or procedural denials to acceptable levels. \rev{To concretely illustrate practical insights, recall that DSS estimated that it needed at least 150 additional staff to reduce wait times below 2 minutes. Our quantitative models estimate that to achieve this \textit{without} any additional staff, DSS must reduce handling time by 75\%, or lengthen recertification periods to 1 year. With just $50$ extra staff, it suffices to reduce handling time by 50\%, and so on.} 

In \Cref{fig-wait-recert,fig:denials-recert} we investigate the same types of staffing curves if we consider \textit{longer re-certification cycles}, from 6 months to 1 or 2 years. While longer re-certification cycles don't improve wait-time efficiency as much as improving service time does, longer recertification periods continue to improve procedural denial efficiency. The moderated joint impact on \textit{wait times} reflects an interesting distinction between \textit{per-individual} performance metrics like average wait time per caller vs. \textit{absolute magnitude} performance metrics like total abandonment/procedural denials. Longer re-certification periods maintain enrollment loads but \textit{reduce} arrival rates, therefore averaging waits over a \textit{smaller} total load and attenuating the improvements from higher throughput. Interventions that reduce {congestion and arrivals} have similar \textit{self-attenuating} effects on rate-based rather than absolute metrics.

\rev{\textit{Broader practical context of potential interventions: }
Different interventions have different institutional constraints: longer recertification periods may require federal policy changes, while reductions in handling time require service redesign that preserves accuracy and quality. 
Practitioners realize improvements via more specific tactics beyond our abstract discussions here \citep{nj-call-centers}. 
While there is also general interest in AI, it cannot completely replace interview provision (due to merit pay designations on who can complete what work), but would be deployed in narrow ways to improve processes.
Our model translates such heterogeneous interventions into key operational inputs, such as arrivals, handling time, and completion rates.
}

\section{\rev{Conclusion}}

\rev{In this work, we introduced a performance evaluation framework for queueing analysis in social services, fit model parameters to call center data from \textit{Holmes v. Knodell}, and derived analytical and practical quantified insights on substitutable improvement levers. 
Since there are no unilateral performance requirements for backlogs in social services, the ultimate goal of our performance evaluation framework is to enable agencies themselves (e.g. via a dashboard) to explore different performance improvement options ex-ante, transparently, and flexibly. 
}

\rev{\textit{Model limitations and future research: } 
The performance metrics we have studied here are of first-order importance operationally, but they are also merely results of the fluid model, which is itself a first-order approximation of the queueing model. Moreover, there are other queueing-theoretic quantities of at least equal practical importance
that cannot be analyzed by a fluid model. For instance, redials are a clear problem in the motivating \emph{Holmes v. Knodell} setting, both in terms of the load on the system from the endogenous congestion and in terms of the administrative burden on individuals. 
A truly probabilistic analysis of the model would allow us to study the distribution of a caller's requisite number of attempts before connection, thus offering insight for how this  emergent arbitrariness might be ameliorated.
}





\paragraph{Ethical Considerations Statement}

The goal of our work was to improve social impact and equity by carefully studying technical systems (queueing) that currently mediate access to crucial benefits, are currently failing, and lack accurate guidance in light of unique characteristics of social service delivery. Our work could have unintended impacts --- after all, we focus on introducing and analyzing performance metrics, but there's always Goodhart's law, ``When a measure becomes a target, it ceases to be a good measure". Indeed, we refer to prior instances of potential gaming or misreporting/misrepresentation of queueing performance metrics, like with efforts to improve the SSA waits \citep{warren}. We try to mitigate potential unintended consequences by highlighting some distinctions between different performance metrics and focusing on a general performance evaluation framework, rather than any single metric in particular. 
\bibliographystyle{ACM-Reference-Format}
\bibliography{FAccT-template/snap-refs.bib} 

\clearpage

\appendix

\section{Background}

\paragraph{Additional discussion on related work}
For the most part, the so-called \textit{civic tech} community is the technical practitioner community focused on digital delivery of services and provision of government benefits. Books and grey literature document the challenges of technical deployment and product design in government settings \citep{pahlka2023recoding}. Recent works in the FAccT community have also explored the interface of benefits technology and equity. \citet{escher2020exposing} audits benefits calculators, which are a category of online tools that distill the government's original, remarkably complex eligibility forms into more UX-friendly, shorter screeners. Such benefits screeners, however, might provide accessibility and reduce frictions at the expense of \textit{accuracy} in eligibility determinations. \citet{jo2025ai} conducts an interview study to assess SNAP applicants' interactions with potential LLM-based chatbots, studying the interface between trust issues in AI and potential improvements in administrative burdens. \citet{koenecke2023popular} find that an efficient advertising budget for GetCalfresh, which streamlines enrollment in California's CalFresh (SNAP), resulted in low takeup among Hispanic individuals: The authors conduct a survey and find general public support for \textit{equitable} advertising budget allocations with higher spending on advertising aimed at harder-to-reach populations.

\paragraph{Additional background on benefits tech}

Government digital infrastructure faces several structural challenges that prevent rapid iteration and performance evaluation feedback loops: contracting out implementation to large vendors in cumbersome public procurement processes, larger "waterfall" and non-"agile" contracts that introduce inflexible system design, limited in-house technical capacity, and vendor opacity due to trade secrets protection. For example, when Missouri announced intentions to deploy a new version of Medicaid home-based care determination algorithms, policy analysts and analysts had to band together to bring potentially affected individuals together to audit potential impacts on care. Benefits technological infrastructure (hence called "benefits tech" for short) falls in this category - it must interface with proprietary data mainframes and is often contracted out via public procurement, where few firms have the capacity to work with government legacy systems, and so there is little market pressure to compete on performance. They found that the updated algorithm would remove 66\% of affected individuals from crucial home nurse care as well as potential syntax errors in the algorithm \citep{noauthor_missouri_nodate}. Without this community audit, such deployments would have introduced harmful erroneous determinations silently --- it would be three more years before Missouri contracted a consultant to analyze the impacts of an updated algorithm. 

The call center system design and deployment in this setting is somewhat different from prior major benefits tech infrastructure projects, in that call center technology is in principle well-developed in the commercial sector. But, call center design and operations are usually contracted out, as Missouri DSS did to Genesys, a major provider. Designing call center systems and queues is a specialized skill with little guidance from FNS as to performance standards and unlikely in-house expertise in social services. 
\clearpage
\section{Dashboard}

\begin{figure}
    \centering
    \includegraphics[width=\linewidth]{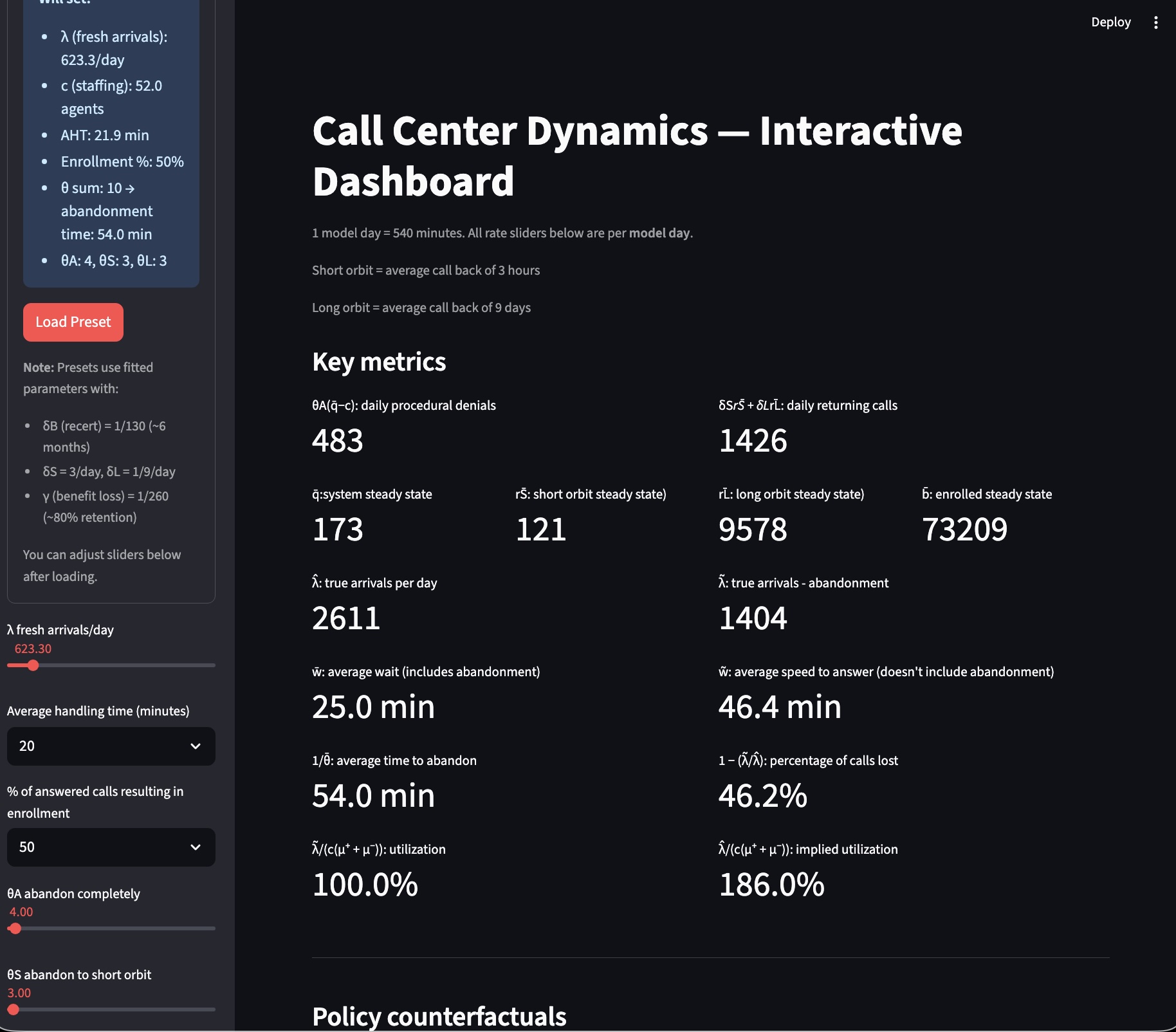}
    \caption{Screenshot of call center dashboard}
    \label{fig:dashboard}
\end{figure}
\clearpage

\section{Additional explanation}\label{apx-addl-discussion}

\paragraph{Additional explanation of Poisson thinning}
For example, consider the dynamics  of the number of presently enrolled  recipients. Intuitively, the next decrease to $B(t)$ will be either a departure via attrition from the benefits system or a return to the call center to re-certify. Because of the memoryless property of the underlying exponential random variables between events in Poisson processes, we can view this as a ``race'' between two exponentially distributed ``clocks'' -- one for attrition and one for re-certification.\footnote{Moreover, because of the thinning property of the Poisson itself, we can granularize this same interpretation to the level of each currently enrolled individual: for each current benefit recipient, there is a $\mathsf{Exp}(\gamma)$ clock for the time to potential attrition and a $\mathsf{Exp}(\delta_\mathsf{B})$ clock for the time to potential re-certification, and whichever clock rings first will be what happens next.} Hence, by these same properties, there is a $\gamma/(\gamma + \delta_\mathsf{B})$  probability that the recipient departs via attrition and a $\delta_\mathsf{B}/(\gamma + \delta_\mathsf{B})$ probability that they instead call to re-certify.
Similar interpretations abound throughout the model: an abandoning caller enters the short orbit with probability $\theta_\mathsf{S}/(\theta_\mathsf{A} + \theta_\mathsf{S} + \theta_\mathsf{L})$, a concluding caller enters the long orbit with probability $\mu_-/(\mu_+ + \mu_-)$, and each of the $(Q(t) \wedge c)$ ongoing calls are equally likely to be the next to conclude.
\section{Call center arrival data }\label{apx-sec-call-center-ts-analysis}

\section*{Data Dictionary: Call Center Performance Metrics (Offered, Answered, Abandoned)}\label{apx-callcenterdata}

\begin{table}[h!]
    \centering
    \caption{Key Metrics from Offered, Answered, Abandoned Data Files}
    \label{tab:data_dictionary_offered_answered}
    \begin{tabular}{|l|p{10cm}|}
        \hline
        \textbf{Field Name} & \textbf{Description} \\
        \hline
        \texttt{Offered} & Total number of incoming calls presented to the call system or queue. \\
        \hline
        \texttt{Answered} & Total number of calls successfully answered by an agent. \\
        \hline
        \texttt{Abandon} & Total number of calls that hung up while waiting in the queue before being answered. \\
        \hline
        \texttt{Abandon \%} & Percentage of calls that abandoned (\texttt{Abandon} / \texttt{Offered}). \\
        \hline
        \texttt{Avg Wait} / \texttt{Average Wait} & Average time a caller spent waiting in the queue before being connected to an agent or abandoning. \\
        \hline
        \texttt{ASA (Average Speed of Answer)} & The average time a call is expected to wait in the queue before being answered by an agent. \\
        \hline
        \texttt{AHT (Average Handle Time)} & The average amount of time an agent spends on an answered call (including talk time and wrap-up work). \\
        \hline
        \texttt{Average Abandon Time} & The average amount of time abandoned calls spent in the queue before the caller hung up. \\
        \hline
        \texttt{Staffing} & The total number of staff members scheduled or logged in to handle calls during the interval. \\
        \hline
        \texttt{Calls/Person} & A staff productivity metric, calculated as calls handled per staff member. \\
        \hline
    \end{tabular}
\end{table}

\subsection{Parameter fitting based on data}


\paragraph{Parameters directly observed in data} 

From the data, we directly observe some of the parameters.



Recall that we have defined $\hat\lambda$ as the total arrival rate of calls, which we observe directly from data. This includes not only the fresh arrivals $\lambda$, i.e. calls originating not from prior cases, but also incoming calls from recertifications $(\delta_{\mathrm{B}} \bar{b})$, re-dials from prior abandoned or re-connected calls entering the short orbit $(\delta_{\mathrm{S}} \bar{r}_{\mathrm{S}})$, and re-dials from prior calls entering the long orbit, $(\delta_{\mathrm{L}} \bar{r}_{\mathrm{L}})$. 

From data, we also directly observe staffing levels, $c,$ for each call center -- we take averages. From SNAP eligibility guidelines, we fix the recertification period is at 6 months (calendar days) $\approx$ 128.5 model (call center) days (since it’s only
open weekdays). This fixes the parameter $\delta_{\mathsf{B}} = 1/128.5$.

The wait time data includes Average Handling Time ($AHT$), 
and therefore the completion rate of calls, $1/AHT$. Our model further distinguishes service completions into those that complete enrollment ($\mu_{\mathsf{+}}$) vs. those that require a re-connection/re-dial ($\mu_{\mathsf{-}}$). Therefore, completion rates $1/AHT$ are only informative of the sum of these enrollment/re-dial flows:
    $$
  \mu_{\mathsf{+}} + \mu_{\mathsf{-}} = \frac{\rev{MMin}}{AHT} $$
  To reduce parameter dimensionality, we fix a proportion $\rho\in[0,1]$ of completed calls that result in completed re/-enrollments, hence $\mu_{\mathsf{+}} = \rho \left( \frac{\rev{MMin}}{AHT} \right).$ Although $\rho$ is not observed in this data, it may be in other settings, or otherwise informed by the domain. 

\paragraph{Parametrizing  $\delta_{\mathsf{L}}$ from the call center arrival data.}
 The re-dial rate of the long orbit is $\delta_{\mathsf{L}}$, which governs the rate at which re-dials occur from the long orbit (longer timeframe of calling back, for example to obtain and confirm additional verifications). 
We fit vector-autoregressive regression (VAR) time-series models to the call center offered and abandoned time series data (adjusting for day-of-week fixed effects). In summary, the empirical time-series analysis surfaces dependence in call center volume (offered and abandoned calls) that operates both over the short-term (dependence on prior 1-2 days) and longer-range dependencies (7-9 days). We therefore set $\delta_{\mathsf{L}} = 1/9$. This is consistent with domain discussion about SNAP interviews and how they often require follow-up due to  interview questions that surface or clarify additional documentation needs. 

\subsection{Call center arrival rate analysis}\label{apx-callcenterarrival-ts}

\begin{table}[htb]
\centering
\caption{Call Center Statistics}
\label{tab:call_center_stats}
\begin{tabular}{ccccccc}
\hline
\textbf{Call Center} & \textbf{Mean Offered} & \textbf{Std } & \textbf{Mean Answered} & \textbf{Std } & \textbf{Mean Abandoned} & \textbf{Std } \\
\hline
2 & 3120.06 & 2573.98 & 1229.37 & 1026.57 & 1788.04 & 1872.56 \\
3 & 2555.30 & 891.27 & 882.08 & 281.89 & 1500.85 & 1020.05 \\
4 & 717.15 & 211.77 & 699.21 & 210.99 & 17.88 & 22.54 \\
5 & 6829.01 & 1943.26 & 2729.42 & 866.85 & 3880.46 & 2054.77 \\
\hline
\end{tabular}
\end{table}


\paragraph{Periodicity}
We start with one call center, the second in pages 37-87 of exhibit 87. 
As expected, there are day of week effects in calls offered and abandoned. Call volume is the highest on Monday, and decreases over the week (for both offered and abandoned calls). 
We adjust for day of week effects by fitting outcomes to a linear model of just day-of-week effects, subtract off the predictions based on just day-of-week and add back the grand mean. 

Next we fit a vector autoregressive model to the call center time-series data of offered and abandoned calls, separately for each call center. For ease of interpretation, we don’t difference the data, although a Dickey-Fuller test finds only the abandoned calls time series is stationary. 

We use 10 lags in the time series analysis. Note that abandoned calls is a proportion of offered calls of the day, and so naturally tracks offered calls. 

To summarize the time series dependence patterns: Both the volume of offered and abandoned calls track 1) recent, past 1-2 days of offered and abandoned calls and 2) mid-longer range dependencies, like 7-10 days prior. The number of offered calls has a significant coefficient on lagged offered calls for at least 7 out of the first 9 lags, i.e. persistence of call volume (but not necessarily abandoned calls). 

The dependence on prior abandonment is intermittent. Overall, abandonment depends more on prior abandonment than the number of offered calls does. But, we also find occasional statistically significant coefficients on long-range (7-9 day) abandonment.

\begin{table}[htbp]
\centering
\caption{Call center 2}
\label{table-callcenter-1}
\small
\renewcommand{\arraystretch}{0.9}
\setlength{\tabcolsep}{4pt}

\begin{tabularx}{0.92\linewidth}{@{\extracolsep{\fill}}>{\raggedright\arraybackslash}Xcc}
\hline\hline
 & \multicolumn{2}{c}{\textit{Dependent variable:}} \\
\cline{2-3}
 & (Offered) & (Abandoned) \\
\hline

Adjusted Offered (Lag 1)  & 0.204$^{***}$ & $-$0.338$^{***}$ \\
 & (0.043) & (0.036) \\

Adjusted Abandon (Lag 1)  & 0.603$^{***}$ & 1.109$^{***}$ \\
 & (0.053) & (0.043) \\

Adjusted Offered (Lag 2)  & 0.541$^{***}$ & 0.154$^{***}$ \\
 & (0.050) & (0.041) \\

Adjusted Abandon (Lag 2)  & $-$0.466$^{***}$ & $-$0.090$^{*}$ \\
 & (0.065) & (0.054) \\

Adjusted Offered (Lag 3)  & 0.091$^{*}$ & 0.123$^{***}$ \\
 & (0.053) & (0.043) \\

Adjusted Abandon (Lag 3)  & $-$0.081 & $-$0.114$^{**}$ \\
 & (0.067) & (0.055) \\

Adjusted Offered (Lag 4)  & 0.114$^{**}$ & 0.069 \\
 & (0.053) & (0.043) \\

Adjusted Abandon (Lag 4)  & $-$0.038 & 0.035 \\
 & (0.067) & (0.055) \\

Adjusted Offered (Lag 5)  & 0.125$^{**}$ & $-$0.024 \\
 & (0.052) & (0.043) \\

Adjusted Abandon (Lag 5)  & $-$0.031 & 0.125$^{**}$ \\
 & (0.067) & (0.055) \\

Adjusted Offered (Lag 6)  & $-$0.121$^{**}$ & $-$0.079$^{*}$ \\
 & (0.053) & (0.043) \\

Adjusted Abandon (Lag 6)  & $-$0.065 & $-$0.071 \\
 & (0.067) & (0.055) \\

Adjusted Offered (Lag 7)  & $-$0.196$^{***}$ & $-$0.016 \\
 & (0.052) & (0.043) \\

Adjusted Abandon (Lag 7)  & 0.264$^{***}$ & 0.061 \\
 & (0.067) & (0.055) \\

Adjusted Offered (Lag 8)  & 0.076 & 0.060 \\
 & (0.053) & (0.043) \\

Adjusted Abandon (Lag 8)  & $-$0.075 & $-$0.053 \\
 & (0.067) & (0.055) \\

Adjusted Offered (Lag 9)  & 0.285$^{***}$ & 0.158$^{***}$ \\
 & (0.052) & (0.043) \\

Adjusted Abandon (Lag 9)  & $-$0.312$^{***}$ & $-$0.186$^{***}$ \\
 & (0.066) & (0.054) \\

Adjusted Offered (Lag 10) & 0.012 & $-$0.086$^{**}$ \\
 & (0.050) & (0.041) \\

Adjusted Abandon (Lag 10) & 0.071 & 0.150$^{***}$ \\
 & (0.064) & (0.053) \\

Adjusted Offered (Lag 11) & $-$0.121$^{***}$ & 0.018 \\
 & (0.044) & (0.036) \\

Adjusted Abandon (Lag 11) & 0.096$^{*}$ & $-$0.039 \\
 & (0.053) & (0.044) \\

const & 31.047 & 5.695 \\
 & (20.605) & (16.878) \\

\hline \\[-1.8ex]
Observations & 2,076 & 2,076 \\
R$^{2}$ & 0.952 & 0.940 \\
Adjusted R$^{2}$ & 0.952 & 0.939 \\
Residual Std. Error (df = 2053) & 562.739 & 460.946 \\
F Statistic (df = 22; 2053) & 1,868.297$^{***}$ & 1,457.608$^{***}$ \\
\hline
\hline \\[-1.8ex]
\textit{Note:} & \multicolumn{2}{r}{$^{*}$p$<$0.1; $^{**}$p$<$0.05; $^{***}$p$<$0.01} \\
\end{tabularx}
\end{table}

\begin{table}[!htbp]
\centering
\caption{Call Center 3}
\label{tbl-call-center-3}
\small
\renewcommand{\arraystretch}{0.88}
\setlength{\tabcolsep}{3.5pt}

\begin{tabularx}{0.92\linewidth}{@{\extracolsep{\fill}}>{\raggedright\arraybackslash}Xcc}
\hline\hline
 & \multicolumn{2}{c}{\textit{Dependent variable:}} \\
\cline{2-3}
 & (Offered) & (Abandoned) \\
\hline

Adjusted Offered (Lag 1)  & $-$0.091 & $-$0.434$^{***}$ \\
 & {\scriptsize (0.083)} & {\scriptsize (0.075)} \\

Adjusted Abandon (Lag 1)  & 0.886$^{***}$ & 1.316$^{***}$ \\
 & {\scriptsize (0.092)} & {\scriptsize (0.083)} \\

Adjusted Offered (Lag 2)  & 0.493$^{***}$ & 0.181$^{**}$ \\
 & {\scriptsize (0.089)} & {\scriptsize (0.081)} \\

Adjusted Abandon (Lag 2)  & $-$0.448$^{***}$ & $-$0.209$^{**}$ \\
 & {\scriptsize (0.117)} & {\scriptsize (0.106)} \\

Adjusted Offered (Lag 3)  & 0.081 & $-$0.009 \\
 & {\scriptsize (0.094)} & {\scriptsize (0.085)} \\

Adjusted Abandon (Lag 3)  & $-$0.167 & $-$0.059 \\
 & {\scriptsize (0.119)} & {\scriptsize (0.108)} \\

Adjusted Offered (Lag 4)  & 0.014 & 0.028 \\
 & {\scriptsize (0.094)} & {\scriptsize (0.085)} \\

Adjusted Abandon (Lag 4)  & 0.048 & 0.014 \\
 & {\scriptsize (0.119)} & {\scriptsize (0.108)} \\

Adjusted Offered (Lag 5)  & 0.123 & $-$0.030 \\
 & {\scriptsize (0.094)} & {\scriptsize (0.085)} \\

Adjusted Abandon (Lag 5)  & 0.025 & 0.230$^{**}$ \\
 & {\scriptsize (0.119)} & {\scriptsize (0.107)} \\

Adjusted Offered (Lag 6)  & 0.047 & 0.108 \\
 & {\scriptsize (0.094)} & {\scriptsize (0.085)} \\

Adjusted Abandon (Lag 6)  & $-$0.161 & $-$0.222$^{**}$ \\
 & {\scriptsize (0.118)} & {\scriptsize (0.107)} \\

Adjusted Offered (Lag 7)  & $-$0.079 & $-$0.011 \\
 & {\scriptsize (0.094)} & {\scriptsize (0.085)} \\

Adjusted Abandon (Lag 7)  & 0.184 & 0.081 \\
 & {\scriptsize (0.119)} & {\scriptsize (0.108)} \\

Adjusted Offered (Lag 8)  & 0.043 & 0.038 \\
 & {\scriptsize (0.093)} & {\scriptsize (0.084)} \\

Adjusted Abandon (Lag 8)  & $-$0.158 & $-$0.152 \\
 & {\scriptsize (0.118)} & {\scriptsize (0.107)} \\

Adjusted Offered (Lag 9)  & 0.317$^{***}$ & 0.230$^{***}$ \\
 & {\scriptsize (0.090)} & {\scriptsize (0.081)} \\

Adjusted Abandon (Lag 9)  & $-$0.270$^{**}$ & $-$0.172 \\
 & {\scriptsize (0.116)} & {\scriptsize (0.104)} \\

Adjusted Offered (Lag 10) & $-$0.0003 & $-$0.108 \\
 & {\scriptsize (0.086)} & {\scriptsize (0.078)} \\

Adjusted Abandon (Lag 10) & 0.053 & 0.129 \\
 & {\scriptsize (0.098)} & {\scriptsize (0.088)} \\

const & 151.875$^{*}$ & 88.717 \\
 & {\scriptsize (78.981)} & {\scriptsize (71.421)} \\

\hline
Observations & 480 & 480 \\
R$^{2}$ & 0.830 & 0.897 \\
Adjusted R$^{2}$ & 0.823 & 0.893 \\
Residual Std. Error (df = 459) & 365.760 & 330.749 \\
F Statistic (df = 20; 459) & 112.110$^{***}$ & 200.132$^{***}$ \\
\hline\hline
\multicolumn{3}{r}{\scriptsize $^{*}$p$<$0.1; $^{**}$p$<$0.05; $^{***}$p$<$0.01} \\
\end{tabularx}
\end{table}

\begin{table}[!htbp]
\centering
\caption{Call Center 4}
\label{tbl-call-center-4}
\small
\renewcommand{\arraystretch}{0.8}
\setlength{\tabcolsep}{3.5pt}

\begin{tabularx}{0.92\linewidth}{@{\extracolsep{\fill}}>{\raggedright\arraybackslash}Xcc}
\hline\hline
 & \multicolumn{2}{c}{\textit{Dependent variable:}} \\
\cline{2-3}
 & Offered & Abandoned \\
\hline

Adjusted Offered (Lag 1)  & 0.453$^{***}$ & 0.012$^{**}$ \\
 & {\scriptsize (0.048)} & {\scriptsize (0.005)} \\

Adjusted Abandon (Lag 1)  & 0.018 & 0.373$^{***}$ \\
 & {\scriptsize (0.478)} & {\scriptsize (0.050)} \\

Adjusted Offered (Lag 2)  & 0.193$^{***}$ & $-$0.024$^{***}$ \\
 & {\scriptsize (0.052)} & {\scriptsize (0.005)} \\

Adjusted Abandon (Lag 2)  & $-$0.421 & 0.044 \\
 & {\scriptsize (0.502)} & {\scriptsize (0.052)} \\

Adjusted Offered (Lag 3)  & 0.071 & 0.001 \\
 & {\scriptsize (0.055)} & {\scriptsize (0.006)} \\

Adjusted Abandon (Lag 3)  & 0.069 & 0.610$^{***}$ \\
 & {\scriptsize (0.508)} & {\scriptsize (0.053)} \\

Adjusted Offered (Lag 4)  & 0.065 & 0.012$^{**}$ \\
 & {\scriptsize (0.054)} & {\scriptsize (0.006)} \\

Adjusted Abandon (Lag 4)  & 0.436 & $-$0.159$^{*}$ \\
 & {\scriptsize (0.924)} & {\scriptsize (0.096)} \\

Adjusted Offered (Lag 5)  & $-$0.017 & $-$0.006 \\
 & {\scriptsize (0.055)} & {\scriptsize (0.006)} \\

Adjusted Abandon (Lag 5)  & 0.357 & 0.066 \\
 & {\scriptsize (0.954)} & {\scriptsize (0.099)} \\

Adjusted Offered (Lag 6)  & $-$0.089 & $-$0.001 \\
 & {\scriptsize (0.055)} & {\scriptsize (0.006)} \\

Adjusted Abandon (Lag 6)  & $-$0.788 & 0.108 \\
 & {\scriptsize (0.963)} & {\scriptsize (0.100)} \\

Adjusted Offered (Lag 7)  & 0.006 & $-$0.010$^{*}$ \\
 & {\scriptsize (0.055)} & {\scriptsize (0.006)} \\

Adjusted Abandon (Lag 7)  & 0.799 & $-$0.062 \\
 & {\scriptsize (0.971)} & {\scriptsize (0.101)} \\

Adjusted Offered (Lag 8)  & 0.065 & 0.018$^{***}$ \\
 & {\scriptsize (0.055)} & {\scriptsize (0.006)} \\

Adjusted Abandon (Lag 8)  & 0.070 & $-$0.218$^{**}$ \\
 & {\scriptsize (0.974)} & {\scriptsize (0.101)} \\

Adjusted Offered (Lag 9)  & 0.113$^{**}$ & 0.018$^{***}$ \\
 & {\scriptsize (0.056)} & {\scriptsize (0.006)} \\

Adjusted Abandon (Lag 9)  & $-$1.014 & $-$0.240$^{**}$ \\
 & {\scriptsize (0.974)} & {\scriptsize (0.101)} \\

Adjusted Offered (Lag 10) & 0.021 & $-$0.009 \\
 & {\scriptsize (0.056)} & {\scriptsize (0.006)} \\

Adjusted Abandon (Lag 10) & 0.403 & $-$0.044 \\
 & {\scriptsize (0.972)} & {\scriptsize (0.101)} \\

Adjusted Offered (Lag 11) & $-$0.063 & $-$0.005 \\
 & {\scriptsize (0.056)} & {\scriptsize (0.006)} \\

Adjusted Abandon (Lag 11) & $-$0.815 & $-$0.080 \\
 & {\scriptsize (0.954)} & {\scriptsize (0.099)} \\

Adjusted Offered (Lag 12) & $-$0.053 & 0.005 \\
 & {\scriptsize (0.056)} & {\scriptsize (0.006)} \\

Adjusted Abandon (Lag 12) & 0.655 & 0.214$^{**}$ \\
 & {\scriptsize (0.946)} & {\scriptsize (0.098)} \\

Adjusted Offered (Lag 13) & 0.016 & $-$0.005 \\
 & {\scriptsize (0.054)} & {\scriptsize (0.006)} \\

Adjusted Abandon (Lag 13) & $-$1.006 & 0.321$^{***}$ \\
 & {\scriptsize (0.738)} & {\scriptsize (0.076)} \\

Adjusted Offered (Lag 14) & 0.085$^{*}$ & $-$0.006 \\
 & {\scriptsize (0.049)} & {\scriptsize (0.005)} \\

Adjusted Abandon (Lag 14) & 0.949 & 0.024 \\
 & {\scriptsize (0.732)} & {\scriptsize (0.076)} \\

const & 105.124$^{***}$ & 1.817 \\
 & {\scriptsize (34.398)} & {\scriptsize (3.561)} \\

\hline
Observations & 476 & 476 \\
R$^{2}$ & 0.531 & 0.570 \\
Adjusted R$^{2}$ & 0.502 & 0.543 \\
Residual Std. Error (df = 447) & 143.892 & 14.896 \\
F Statistic (df = 28; 447) & 18.106$^{***}$ & 21.184$^{***}$ \\
\hline\hline
\multicolumn{3}{r}{\scriptsize $^{*}$p$<$0.1; $^{**}$p$<$0.05; $^{***}$p$<$0.01} \\
\end{tabularx}
\end{table}

\begin{table}[!htbp]
\centering
\caption{Call Center 5}
\label{tbl-call-center-5}
\small
\renewcommand{\arraystretch}{0.88}
\setlength{\tabcolsep}{3.5pt}

\begin{tabularx}{0.92\linewidth}{@{\extracolsep{\fill}}>{\raggedright\arraybackslash}Xcc}
\hline\hline
 & \multicolumn{2}{c}{\textit{Dependent variable:}} \\
\cline{2-3}
 & (Offered) & (Abandoned) \\
\hline

Adjusted Offered (Lag 1)  & 0.017 & $-$0.402$^{***}$ \\
 & {\scriptsize (0.099)} & {\scriptsize (0.098)} \\

Adjusted Abandon (Lag 1)  & 0.754$^{***}$ & 1.183$^{***}$ \\
 & {\scriptsize (0.102)} & {\scriptsize (0.100)} \\

Adjusted Offered (Lag 2)  & 0.586$^{***}$ & 0.216$^{**}$ \\
 & {\scriptsize (0.107)} & {\scriptsize (0.105)} \\

Adjusted Abandon (Lag 2)  & $-$0.489$^{***}$ & $-$0.126 \\
 & {\scriptsize (0.119)} & {\scriptsize (0.117)} \\

Adjusted Offered (Lag 3)  & 0.151 & 0.142 \\
 & {\scriptsize (0.113)} & {\scriptsize (0.111)} \\

Adjusted Abandon (Lag 3)  & $-$0.196 & $-$0.212$^{*}$ \\
 & {\scriptsize (0.124)} & {\scriptsize (0.122)} \\

Adjusted Offered (Lag 4)  & 0.027 & $-$0.050 \\
 & {\scriptsize (0.113)} & {\scriptsize (0.111)} \\

Adjusted Abandon (Lag 4)  & 0.044 & 0.140 \\
 & {\scriptsize (0.123)} & {\scriptsize (0.122)} \\

Adjusted Offered (Lag 5)  & $-$0.219$^{*}$ & $-$0.240$^{**}$ \\
 & {\scriptsize (0.113)} & {\scriptsize (0.111)} \\

Adjusted Abandon (Lag 5)  & 0.269$^{**}$ & 0.341$^{***}$ \\
 & {\scriptsize (0.123)} & {\scriptsize (0.121)} \\

Adjusted Offered (Lag 6)  & 0.064 & 0.038 \\
 & {\scriptsize (0.111)} & {\scriptsize (0.109)} \\

Adjusted Abandon (Lag 6)  & $-$0.152 & $-$0.129 \\
 & {\scriptsize (0.122)} & {\scriptsize (0.120)} \\

Adjusted Offered (Lag 7)  & $-$0.021 & 0.108 \\
 & {\scriptsize (0.108)} & {\scriptsize (0.107)} \\

Adjusted Abandon (Lag 7)  & 0.114 & $-$0.060 \\
 & {\scriptsize (0.118)} & {\scriptsize (0.116)} \\

Adjusted Offered (Lag 8)  & 0.004 & 0.002 \\
 & {\scriptsize (0.100)} & {\scriptsize (0.099)} \\

Adjusted Abandon (Lag 8)  & $-$0.046 & $-$0.001 \\
 & {\scriptsize (0.111)} & {\scriptsize (0.110)} \\

Adjusted Offered (Lag 9)  & 0.355$^{***}$ & 0.225$^{**}$ \\
 & {\scriptsize (0.092)} & {\scriptsize (0.090)} \\

Adjusted Abandon (Lag 9)  & $-$0.331$^{***}$ & $-$0.224$^{**}$ \\
 & {\scriptsize (0.096)} & {\scriptsize (0.094)} \\

const & 343.921$^{**}$ & 78.282 \\
 & {\scriptsize (168.865)} & {\scriptsize (166.516)} \\

\hline
Observations & 481 & 481 \\
R$^{2}$ & 0.876 & 0.896 \\
Adjusted R$^{2}$ & 0.871 & 0.892 \\
Residual Std. Error (df = 462) & 673.272 & 663.906 \\
F Statistic (df = 18; 462) & 180.691$^{***}$ & 220.763$^{***}$ \\
\hline\hline
\multicolumn{3}{r}{\scriptsize $^{*}$p$<$0.1; $^{**}$p$<$0.05; $^{***}$p$<$0.01} \\
\end{tabularx}
\end{table}

\clearpage
\section{Analysis}

\begin{lemma}
Suppose the system is in the overloaded regime. Then, the total arrival rate and the total arrival rate without abandonments are given by
\begin{align}
\hat \lambda
&=
c\mu
+
\frac{\theta}{\theta_\mathsf{A}}\left(\lambda - \frac{\gamma}{\gamma + \delta_\mathsf{B}} c\mu_+\right)
\qquad\quad
\text{  and  }
\qquad\qquad
\tilde\lambda
=
c\mu
,
\end{align}
respectively.
\end{lemma}
\begin{proof}
These arrival rate identities follow immediately from the expressions for the steady-state fluid values in the overloaded regime from Corollary~\ref{overCor} and the definitions that $\hat \lambda = \lambda + \delta_\mathsf{B} \bar b + \delta_\mathsf{S} \bar r_\mathsf{S} + \delta_\mathsf{L} \bar r_\mathsf{L}$ and $\tilde \lambda = \hat \lambda - \theta(\bar q - c)$.
\end{proof}

\begin{lemma}\label{rateLemma}
Suppose the system is in the overloaded regime. 
Phrased in terms of the procedural denial rate ($\mathrm{PD}$), the mean number of waiting callers ($\bar{q} - c$), mean waiting time ($\bar w$), average speed to answer ($\tilde w$), the endogenous congestion from re-dials ($\mathrm{EC}_\mathsf{R}$), and the endogenous congestion from re-certification ($\mathrm{EC}_\mathsf{B}$) can be expressed
\begin{align}
\bar q - c 
&= 
\frac{1}{\theta_\mathsf{A}}\mathrm{PD}
,
\\
\bar w
&=
\frac{\mathrm{PD}}{c\mu\theta_\mathsf{A} + \theta \mathrm{PD}}
,
\\
\tilde w
&=
\frac{1}{c\mu \theta_\mathsf{A} } \mathrm{PD}
,
\\
\mathrm{EC}_\mathsf{R}
&=
c \mu_-
+
\frac{\theta_\mathsf{S} + \theta_\mathsf{L}}{\theta_\mathsf{A}}
\mathrm{PD}
,
\\
\mathrm{EC}_\mathsf{B}
&=
c\mu_+ - \lambda + \mathrm{PD}
.
\end{align}
\end{lemma}
\begin{proof}
First, let us recall that, in the overloaded regime, the procedural denial rate is given by
\begin{align}
\mathrm{PD}
&=
\lambda - \frac{\gamma}{\gamma + \delta_\mathsf{B}} c\mu_+
.
\end{align}
Hence, by immediate consequence of Lemma~\ref{rateLemma}, the total arrival rate can be phrased in terms of the procedural denials as
\begin{align}
\hat \lambda 
&=
c\mu + \frac{\theta}{\theta_\mathsf{A}} \mathrm{PD}
.
\end{align}
The proofs of the expressions for the five performance metrics now quickly follow. First, for the mean number of waiting callers, we can observe
\begin{align}
\bar q - c
&=
\frac{1}{\theta_\mathsf{A}}
\left(
\lambda - \frac{\gamma}{\gamma + \delta_\mathsf{B}}c\mu_+
\right)
=
\frac{1}{\theta_\mathsf{A}}\mathrm{PD}
.
\end{align}
Accordingly, the waiting time and speed to answer can be written
\begin{align}
\bar w 
&=
\frac{\bar q - c}{\hat \lambda}
=
\frac{\frac{1}{\theta_\mathsf{A}}\mathrm{PD}}
{c\mu + \frac{\theta}{\theta_\mathsf{A}}\mathrm{PD}}
,
\end{align}
and
\begin{align}
\tilde w
&=
\frac{\bar q - c}{\tilde \lambda}
=
\frac{\frac{1}{\theta_\mathsf{A}}\mathrm{PD}}{c\mu}
,
\end{align}
which both immediately simplify to the stated expressions. 
Finally, for the endogenous congestion terms, we can draw upon the equilibrium solutions for the overloaded regime in Corollary~\ref{overCor} and find
\begin{align}
\mathrm{EC}_\mathsf{R}
&=
\delta_\mathsf{S}\bar r_\mathsf{S} + \delta_\mathsf{L}\bar r_\mathsf{L}
=
c\mu_-
+ 
\frac{\theta_\mathsf{S} + \theta_\mathsf{L}}{\theta_\mathsf{A}}\left( \lambda - \frac{\gamma}{\gamma + \delta_\mathsf{B}}c\mu_+\right)
=
c\mu_-
+ 
\frac{\theta_\mathsf{S} + \theta_\mathsf{L}}{\theta_\mathsf{A}}
\mathrm{PD}
,
\end{align}
and
\begin{align}
\mathrm{EC}_\mathsf{B}
&=
\delta_\mathsf{B}\bar b
=
\frac{\delta_\mathsf{B}}{\gamma + \delta_\mathsf{B}} c\mu_+
=
c\mu_+
-
\frac{\gamma}{\gamma + \delta_\mathsf{B}}
c\mu_+
=
c\mu_+
+
\mathrm{PD}
-
\lambda
,
\end{align}
which completes the proof.
\end{proof}

\begin{proposition}\label{dPDprop}
Suppose the system is in the overloaded regime. For the parameters of operational design and control, $\delta_\mathsf{B}$, $c$, $\mu$, and $p_+$ where $\mu_+ = \mu p_+$ and $\mu_- = \mu(1-p_+)$, the partial derivatives of the procedural denial rate are
\begin{align}
\frac{\partial \mathrm{PD}}{\partial \delta_\mathsf{B}}
&=
\frac{\gamma c\mu p_+ }{(\gamma + \delta_\mathsf{B})^2}
,
\\
\frac{\partial \mathrm{PD}}{\partial c}
&=
-
\frac{\gamma \mu p_+}{\gamma + \delta_\mathsf{B}}
,
\\
\frac{\partial \mathrm{PD}}{\partial \mu}
&=
-
\frac{\gamma c p_+}{\gamma + \delta_\mathsf{B}}
,
\\
\frac{\partial \mathrm{PD}}{\partial p_+}
&=
-
\frac{\gamma c \mu}{\gamma + \delta_\mathsf{B}}
.
\end{align}
\end{proposition}
\begin{proof}
Each of these derivatives immediately follows from the expression of the procedural denial rate in terms of the four operational parameters, namely 
$\mathrm{PD}
=
\lambda
-
{\gamma c \mu p_+}/({\gamma + \delta_\mathsf{B}})$.
\end{proof}

\begin{proposition}\label{qNumWaitprop}
Suppose the system is in the overloaded regime. For the parameters of operational design and control, $\delta_\mathsf{B}$, $c$, $\mu$, and $p_+$ where $\mu_+ = \mu p_+$ and $\mu_- = \mu(1-p_+)$, the partial derivatives of the mean number of waiting callers can be expressed
\begin{align}
\frac{\partial (\bar q - c)}{\partial \delta_\mathsf{B}}
&=
\frac{\gamma c\mu p_+ }{\theta_\mathsf{A}(\gamma + \delta_\mathsf{B})^2}
,
\\
\frac{\partial(\bar q - c)}{\partial c}
&=
-
\frac{\gamma \mu p_+}{\theta_\mathsf{A}(\gamma + \delta_\mathsf{B})}
,
\\
\frac{\partial (\bar q - c)}{\partial \mu}
&=
-
\frac{\gamma c p_+}{\theta_\mathsf{A}(\gamma + \delta_\mathsf{B})}
,
\\
\frac{\partial (\bar q - c)}{\partial p_+}
&=
-
\frac{\gamma c \mu}{\theta_\mathsf{A}(\gamma + \delta_\mathsf{B})}
\end{align}
\end{proposition}
\begin{proof}
These expressions follow immediately from Lemma~\ref{rateLemma} and Proposition~\ref{dPDprop}, as $(\bar q - c)$ is linear in $\mathrm{PD}$.
\end{proof}

\begin{proposition}\label{dWaitprop}
Suppose the system is in the overloaded regime. For the parameters of operational design and control, $\delta_\mathsf{B}$, $c$, $\mu$, and $p_+$ where $\mu_+ = \mu p_+$ and $\mu_- = \mu(1-p_+)$, the partial derivatives of the mean waiting time can be expressed
\begin{align}
\frac{\partial \bar w}{\partial \delta_\mathsf{B}}
&=
\gamma \theta_\mathsf{A} p_+
\left(
\frac{c\mu }{\left(c\mu\theta_\mathsf{A} + \theta \mathrm{PD}\right)(\gamma + \delta_\mathsf{B})}
\right)^2
,
\\
\frac{\partial \bar w}{\partial c}
&=
-
\frac{\lambda \mu \theta_\mathsf{A} }{\left(c\mu\theta_\mathsf{A} + \theta \mathrm{PD}\right)^2}
,
\\
\frac{\partial \bar w}{\partial \mu}
&=
-
\frac{\lambda c \theta_\mathsf{A} }{\left(c\mu\theta_\mathsf{A} + \theta \mathrm{PD}\right)^2}
,
\\
\frac{\partial \bar w}{\partial p_+}
&=
-
\frac{\gamma \theta_\mathsf{A}}{\gamma + \delta_\mathsf{B}}
\left(
\frac{c\mu}{c\mu\theta_\mathsf{A} + \theta \mathrm{PD}}
\right)^2
\end{align}
\end{proposition}
\begin{proof}
By Lemma~\ref{rateLemma}, for each $x \in \{\delta_\mathsf{B}, c, \mu, p_+\}$, the partial derivative of the mean waiting time can be related to the partial derivative of the procedural denial rate via
\begin{align}
\frac{\partial \bar w}{\partial x}
&=
\frac{\partial}{\partial x}
\left(
\frac{\mathrm{PD}}{c\mu\theta_\mathsf{A} + \theta \mathrm{PD}}
\right)
\\
&=
\frac{1}{c\mu\theta_\mathsf{A} + \theta \mathrm{PD}}
\frac{\partial \mathrm{PD}}{\partial x}
-
\frac{\mathrm{PD}}{\left(c\mu\theta_\mathsf{A} + \theta \mathrm{PD}\right)^2}
\left(
\theta_\mathsf{A} \frac{\partial (c\mu)}{\partial x}  + \theta \frac{\partial \mathrm{PD}}{\partial x}
\right)
\\
&=
\frac{c\mu\theta_\mathsf{A}}{\left(c\mu\theta_\mathsf{A} + \theta \mathrm{PD}\right)^2}
\frac{\partial \mathrm{PD}}{\partial x}
-
\frac{\theta_\mathsf{A} \mathrm{PD}}{\left(c\mu\theta_\mathsf{A} + \theta \mathrm{PD}\right)^2}
 \frac{\partial (c\mu)}{\partial x}  
 \\
&=
\frac{\theta_\mathsf{A} }{\left(c\mu\theta_\mathsf{A} + \theta \mathrm{PD}\right)^2}
\left(
c\mu
\frac{\partial \mathrm{PD}}{\partial x}
-
\mathrm{PD}
 \frac{\partial (c\mu)}{\partial x}  
 \right)
 .
\end{align}
Hence, by consequence of the partial derivations of $\mathrm{PD}$ provided in Proposition~\ref{dPDprop}, we immediately obtain the partial derivatives of $\bar w$.
\end{proof}

\begin{proposition}\label{dASAprop}
Suppose the system is in the overloaded regime. For the parameters of operational design and control, $\delta_\mathsf{B}$, $c$, $\mu$, and $p_+$ where $\mu_+ = \mu p_+$ and $\mu_- = \mu(1-p_+)$, the partial derivatives of the average speed to answer can be expressed
\begin{align}
\frac{\partial \tilde w}{\partial \delta_\mathsf{B}}
&=
\frac{\gamma p_+ }{\theta_\mathsf{A}(\gamma + \delta_\mathsf{B})^2}
,
\\
\frac{\partial \tilde w}{\partial c}
&=
-
\frac{\gamma p_+}{c\theta_\mathsf{A}(\gamma + \delta_\mathsf{B})}
-
\frac{\mathrm{PD}}{c^2 \mu \theta_\mathsf{A}}
,
\\
\frac{\partial \tilde w}{\partial \mu}
&=
-
\frac{\gamma p_+}{\mu \theta_\mathsf{A}(\gamma + \delta_\mathsf{B})}
-
\frac{\mathrm{PD}}{c\mu^2\theta_\mathsf{A}}
,
\\
\frac{\partial \tilde w}{\partial p_+}
&=
-
\frac{\gamma }{\theta_\mathsf{A}(\gamma + \delta_\mathsf{B})}
.
\end{align}
\end{proposition}
\begin{proof}
From Lemma~\ref{rateLemma}, we have that for $x \in \{c,\mu\}$, 
\begin{align}
\frac{\partial \tilde w}{\partial x}
&=
\frac{1}{c\mu \theta_\mathsf{A}} \frac{\partial \mathrm{PD}}{\partial x}
+
\frac{\mathrm{PD}}{\theta_\mathsf{A}}
\frac{\partial}{\partial x}
\left(
\frac{1}{c\mu}
\right)
,
\end{align}
and, for $x \in \{\delta_\mathsf{B}, p_+\}$, the partial derivative of $\tilde w$ is simply that of $\mathrm{PD}$ divided by $c\mu \theta_\mathsf{A}$. Hence, via Proposition~\ref{dPDprop}, we achieve the stated expressions.
\end{proof}

\begin{proposition}\label{dECRprop}
Suppose the system is in the overloaded regime. For the parameters of operational design and control, $\delta_\mathsf{B}$, $c$, $\mu$, and $p_+$ where $\mu_+ = \mu p_+$ and $\mu_- = \mu(1-p_+)$, the partial derivatives of the endogenous congestion from re-dials can be expressed
\begin{align}
\frac{\partial \mathrm{EC}_\mathsf{R}}{\partial \delta_\mathsf{B}}
&=
\frac{\gamma c\mu p_+ (\theta_\mathsf{S} + \theta_\mathsf{L}) }{\theta_\mathsf{A}(\gamma + \delta_\mathsf{B})^2}
,
\\
\frac{\partial \mathrm{EC}_\mathsf{R}}{\partial c}
&=
 \mu (1-p_+)
-
\frac{\gamma \mu p_+ (\theta_\mathsf{S} + \theta_\mathsf{L}) }
{\theta_\mathsf{A}(\gamma + \delta_\mathsf{B})}
,
\\
\frac{\partial \mathrm{EC}_\mathsf{R}}{\partial \mu}
&=
c (1-p_+)
-
\frac{\gamma c p_+ (\theta_\mathsf{S} + \theta_\mathsf{L}) }
{\theta_\mathsf{A}(\gamma + \delta_\mathsf{B})}
,
\\
\frac{\partial \mathrm{EC}_\mathsf{R}}{\partial p_+}
&=
-c \mu 
-
\frac{\gamma c \mu (\theta_\mathsf{S} + \theta_\mathsf{L}) }
{\theta_\mathsf{A}(\gamma + \delta_\mathsf{B})}
\end{align}
\end{proposition}
\begin{proof}
By Lemma~\ref{rateLemma}, the partial derivative of the re-dial endogenous congestion is
\begin{align}
\frac{\partial \mathrm{EC}_\mathsf{R}}{\partial x}
&=
\frac{\partial}{\partial x}\left(c \mu (1-p_+)\right)
+
\frac{\theta_\mathsf{S} + \theta_\mathsf{L}}{\theta_\mathsf{A}}
\frac{\partial \mathrm{PD}}{\partial x}
\end{align}
for each $x \in \{\delta_\mathsf{B}, c, \mu, p_+\}$. Hence, by Propositon~\ref{dPDprop}, we immediately achieve the stated results.
\end{proof}

\begin{proposition}\label{dECBprop}
Suppose the system is in the overloaded regime. For the parameters of operational design and control, $\delta_\mathsf{B}$, $c$, $\mu$, and $p_+$ where $\mu_+ = \mu p_+$ and $\mu_- = \mu(1-p_+)$, the partial derivatives of the endogenous congestion from re-certifications can be expressed
\begin{align}
\frac{\partial \mathrm{EC}_\mathsf{B}}{\partial \delta_\mathsf{B}}
&=
\frac{\gamma c\mu p_+ }{(\gamma + \delta_\mathsf{B})^2}
,
\\
\frac{\partial \mathrm{EC}_\mathsf{B}}{\partial c}
&=
\frac{\delta_\mathsf{B} \mu p_+}{\gamma + \delta_\mathsf{B}}
,
\\
\frac{\partial \mathrm{EC}_\mathsf{B}}{\partial \mu}
&=
\frac{\delta_\mathsf{B} c p_+}{\gamma + \delta_\mathsf{B}}
,
\\
\frac{\partial \mathrm{EC}_\mathsf{B}}{\partial p_+}
&=
\frac{\delta_\mathsf{B} c \mu}{\gamma + \delta_\mathsf{B}}
.
\end{align}
\end{proposition}
\begin{proof}
Lemma~\ref{rateLemma} implies that the partial derivative of the re-certification endogenous congestion is simply a sum of two partial derivatives,
\begin{align}
\frac{\partial \mathrm{EC}_\mathsf{B}}{\partial x}
&=
\frac{\partial}{\partial x}\left(c\mu p_+\right) + \frac{\partial \mathrm{PD}}{\partial x}
,
\end{align}
and from this relationship, we immediately achieve the proof via simplifying from Proposition~\ref{dPDprop}.
\end{proof}

\end{document}